\documentclass[conference]{IEEEtran}

\usepackage{amsthm}
\usepackage{amsmath}
\usepackage{amsfonts}
\usepackage{mathtools}
\usepackage{float}
\usepackage{colortbl}
\usepackage{framed}
\usepackage{enumitem}
\usepackage{tikz}
\usepackage{multirow}
\usepackage{balance}
\usepackage{enumitem}
\usepackage{caption}
\usepackage{algorithm}
\usepackage{hyperref}
\usepackage{tikz}
\usepackage{graphicx}
\usepackage{subcaption}
\usepackage{extarrows}
\usepackage{turnstile}
\usepackage{wrapfig}
\usepackage{arydshln}
\usepackage{pbox}
\def\BibTeX{{\rm B\kern-.05em{\sc i\kern-.025em b}\kern-.08em
    T\kern-.1667em\lower.7ex\hbox{E}\kern-.125emX}}

\hypersetup{
    colorlinks=true,
    linkcolor=blue,
    filecolor=magenta,      
    urlcolor=cyan,
}

\newtheorem{definition}{Definition}
\newtheorem{assumption}{Assumption}
\newtheorem{theorem}{Theorem}
\newtheorem{lemma}{Lemma}
\newtheorem{fact}{Fact}

\providecommand{\citep}{\cite}
\newcommand{\mycomment}[1]{{\color{gray}// #1}}

\newcommand{\ve}[1]{\boldsymbol{#1}}  
\newcommand{\vc}[2]{{#1}_{#2}}  
\newcommand{\si}[2]{#1 [\![ #2 ]\!] }  
\newcommand{\so}[2]{#1 [\![ #2 ]\!] }  
\newcommand{\multiset}[1]{\ensuremath{ \{\!\!\{#1\}\!\!\} }}

\newcommand{\ibooth}{\ensuremath{j}}  
\newcommand{\nbooths}{\ensuremath{\ell}}  
  
\newcommand{\adv}{\ensuremath{\mathcal{A}}}  

\newcommand{\setup}{\ensuremath{\mathsf{Setup}}}        
\newcommand{\register}{\ensuremath{\mathsf{Register}}}  
\newcommand{\prepareRoll}{\ensuremath{\mathsf{PrepER}}} 
\newcommand{\cast}{\ensuremath{\mathsf{Cast}}}          
\newcommand{\indrraudit}{\ensuremath{\mathsf{IndRAudit}}} 
\newcommand{\indcraudit}{\ensuremath{\mathsf{IndCAudit}}} 
\newcommand{\univaudit}{\ensuremath{\mathsf{UnivAudit}}}
\newcommand{\process}{\ensuremath{\mathsf{Process}}}    
\newcommand{\gset}{\ensuremath{\mathcal{O}}}            
\newcommand{\gelg}{\ensuremath{\mathcal{O}_{\mathsf{elg}}}}       
\newcommand{\gcapture}{\ensuremath{\mathcal{O}_{\mathsf{capture}}}} 
\newcommand{\glive}{\ensuremath{\mathcal{O}_{\mathsf{live}}}} 
\newcommand{\gmatch}{\ensuremath{\mathcal{O}_{\mathsf{match}}}}   
\newcommand{\guid}{\ensuremath{\mathcal{O}_{\mathsf{uid}}}}       
\newcommand{\pietoev}{\ensuremath{\Pi_{\mathsf{E2E}\text{-}\mathsf{V}}}} 
\newcommand{\ocast}{\ensuremath{\mathsf{OCast}}} 
\newcommand{\oreg}{\ensuremath{\mathsf{OReg}}} 

\newcommand{\RER}{\mathsf{ERR}}  
\newcommand{\ERR}{\mathsf{ERR}}  
\newcommand{\ER}{\mathsf{ER}}  
\newcommand{\CI}{\mathsf{CI}}  
\newcommand{\EV}{\mathsf{EV}}  
\newcommand{\EVh}{\EV_{\mathsf{h}}}  
\newcommand{\uid}{\ensuremath{\mathsf{uid}}}  
\newcommand{\vid}{\ensuremath{\mathsf{vid}}}  
\newcommand{\bidR}{\ensuremath{\mathsf{rp}}}  
\newcommand{\bidpreR}{\ensuremath{\mathsf{prp}}}  
\newcommand{\bidC}{\ensuremath{\mathsf{cp}}}  
\newcommand{\bidpreC}{\ensuremath{\mathsf{pcp}}}  
\newcommand{\ballot}[2]{\ensuremath{\vc{b}{(#1, #2)}}}  
\newcommand{\iballot}{\ensuremath{\beta}}  
\newcommand{\nballots}[1]{\ensuremath{n_{\mathsf{b}_{#1}}}}  
\newcommand{\ballotsbooth}[1]{\ensuremath{(\ballot{\iballot}{#1})_{\iballot \in [\nballots{#1}]}}}  
\newcommand{\ballots}[1]{\ensuremath{(\ballotsbooth{\ibooth})_{\ibooth \in [\nbooths]}}}  
\newcommand{\ev}[1]{\ensuremath{\vc{c}{#1}}}  
\newcommand{\ed}{\ensuremath{\mathsf{ed}}}  
\newcommand{\rrec}{\ensuremath{\mathsf{rr}}}  
\newcommand{\rvalid}{\ensuremath{\mathsf{rvalid}}}  
\newcommand{\rresp}{\ensuremath{\mathsf{elg}}}  
\newcommand{\crec}{\ensuremath{\mathsf{cr}}}  
\newcommand{\cvalid}{\ensuremath{\mathsf{cvalid}}}  
\newcommand{\ver}{\ensuremath{\mathsf{ver}}}  
\newcommand{\pk}{\ensuremath{\mathsf{pk}}}  
\newcommand{\sk}{\ensuremath{\mathsf{sk}}}  
\newcommand{\UA}{\ensuremath{\mathsf{UA}}}  

\newcommand{\nizk}{\ensuremath{\mathsf{NIZK}}}

\newcommand{\nizkver}{\ensuremath{\mathsf{NIZKVer}}}  

\newcommand{\Nat}{\mathbb{N}}  
\newcommand{\Voters}{\ensuremath{\mathbb{V}}}  

\newcommand{\keygen}{\ensuremath{\mathsf{Keygen}}}  
\newcommand{\enc}{\ensuremath{\mathsf{Enc}}}  
\newcommand{\encscheme}{\ensuremath{\mathsf{E}}}
\newcommand{\renc}{\ensuremath{\mathsf{REnc}}}

\newcommand{\ths}{\ensuremath{\encscheme^{\mathsf{th}}}}
\newcommand{\thsg}{\ensuremath{\ths.\keygen}}
\newcommand{\thse}{\ensuremath{\ths.\enc}}
\newcommand{\thsre}{\ensuremath{\ths.\renc}}
\newcommand{\tdec}{\ensuremath{\mathsf{TDec}}}
\newcommand{\thsd}{\ensuremath{\ths.\tdec}}

\newcommand{\h}{\ensuremath{\mathsf{h}}}
\newcommand{\rh}{\ensuremath{\mathsf{rh}}}
\newcommand{\rdh}{\ensuremath{\mathsf{rdh}}}
\newcommand{\ch}{\ensuremath{\mathsf{ch}}}
\newcommand{\cdh}{\ensuremath{\mathsf{cdh}}}

\newcommand{\mult}[2]{\ensuremath{\theta_{#1}(#2)}}

\newcommand{\valid}{\ensuremath{\mathsf{Valid}}}

\newcommand{\accept}{\ensuremath{``\mathsf{Accept}"}}
\newcommand{\reject}{\ensuremath{``\mathsf{Reject}"}}
\newcommand{\alreadyreg}{\ensuremath{``\mathsf{DupReg}"}}
\newcommand{\alreadycast}{\ensuremath{``\mathsf{DupCast}"}}
\newcommand{\env}{\ensuremath{\mathsf{env}}}
\newcommand{\envR}{\ensuremath{\env_{\mathsf{r}}}}
\newcommand{\envC}{\ensuremath{\env_{\mathsf{c}}}}

\newcommand{\envpreR}{\ensuremath{\env_{\mathsf{pr}}}}
\newcommand{\envpreC}{\ensuremath{\env_{\mathsf{pc}}}}

\renewcommand{\ev}{\ensuremath{\mathsf{ev}}}
\renewcommand{\r}{\ensuremath{\mathsf{r}}}
\renewcommand{\c}{\ensuremath{\mathsf{c}}}
\renewcommand{\d}{\ensuremath{\mathsf{d}}}
\newcommand{\s}{\ensuremath{\mathsf{s}}}

\begin{document}

\title{Publicly auditable privacy-preserving electoral rolls}
\author{\IEEEauthorblockN{Prashant Agrawal\IEEEauthorrefmark{1}\IEEEauthorrefmark{2}} \and \IEEEauthorblockN{Mahabir Prasad Jhanwar\IEEEauthorrefmark{2}} \and \IEEEauthorblockN{Subodh Vishnu Sharma\IEEEauthorrefmark{1}} \and \IEEEauthorblockN{Subhashis Banerjee\IEEEauthorrefmark{2}} \and
\IEEEauthorblockA{\IEEEauthorrefmark{1}Department of Computer Science and Engineering \\ Indian Institute of Technology Delhi \\ New Delhi, India \\ \texttt{\footnotesize \{prashant,svs\}@cse.iitd.ac.in}} \and
\IEEEauthorblockA{\IEEEauthorrefmark{2}Department of Computer Science  and \\ Center for Digitalisation, AI and Society \\ Ashoka University \\ Sonipat, India \\ \texttt{\footnotesize \{mahavir.jhawar,suban\}@ashoka.edu.in}}
}

\maketitle

\begin{abstract}
	While existing literature on electronic voting has extensively addressed verifiability of voting protocols, the vulnerability of electoral rolls in large public elections remains a critical concern. To ensure integrity of electoral rolls, the current practice is to either make electoral rolls public or share them with the political parties. However, this enables  construction of detailed voter profiles and selective targeting and manipulation of voters, thereby undermining the fundamental principle of free and fair elections. In this paper, we study the problem of designing publicly auditable yet privacy-preserving electoral rolls. We first formulate a threat model and provide formal security definitions. We then present a protocol for creation, maintenance and usage of electoral rolls that mitigates the threats. Eligible voters can verify their inclusion, whereas political parties and auditors can statistically audit the electoral roll. Further, the audit can also detect polling-day ballot stuffing and denials to eligible voters by malicious polling officers. The entire electoral roll is never revealed, which prevents any large-scale systematic voter targeting and manipulation.
\end{abstract}


\section{Introduction}
\label{sec:intro}

Most of the research in electronic voting have concentrated on end-to-end verifiable voting protocols \citep{pretavoter,scantigrity,punchscan,scratchandvote,bingovoting,starvote} which safeguard against various attacks on the vote recording and counting systems. However, it is relatively simpler for an attacker to instead target the  \emph{electoral roll}, the official list of eligible voters. Unlike the tightly regulated polling process, the voter registration process for creating the electoral roll spans several months, involves multiple distributed entities working on the ground, and is generally less protected, making it much easier to attack. Incidences and allegations of both administrative errors and active manipulation of the voter list --- including addition of dead people, minors, non-citizens and even ``phantom'' voters, malicious removal of ``unwanted'' eligible voters, and duplicate entries for some eligible voters --- are all frequently reported problems \citep{election-rigging-how-to-fight,elections-without-democracy,how-to-rig-an-election,retnakumar-electoral-rolls-india,kodali-electoralroll-deletion,bhatnagar-electoralroll-deletion,special-issue-electoral-fraud-india-pak}.

A common approach to engender public confidence in the integrity of electoral rolls is to make them public \citep{RER60} to enable public audit and grievance redressal. However, this introduces a wide array of privacy issues. The electoral roll  typically includes fields like name, gender, date of birth, photo, address, unique identifiers such as voter ID  or social security numbers, and, in some cases, even party affiliations of voters \citep{howard-kreiss-survey-aus-can-uk-us}. As many political and social observers  have highlighted \citep{onselen-electoral-databases,errington-suiting-themselves,howard-kreiss-survey-aus-can-uk-us,special-issue-electoral-fraud-india-pak,hunter,voter-privacy-big-data}, voters lists are routinely utilised  to target and manipulate voters, particularly in closely contested elections.  Voters lists often have enough information to guess voters' voting preferences, either directly if the party affiliation is  included, or indirectly  through  proxies such as religion or social identities predicted using voter names \citep{chaturvedi-religion-predictor}, or income predicted using addresses. This, along with the fact that the electoral roll data can be linked to other available databases, often via the unique identifiers, allows construction of  even more detailed voter profiles to identify strong supporters, strong opposers and swing voters,  economically and socially vulnerable voters, and individual voters' lifestyle preferences and  interests. As an example, the UK Conservative Party claimed to achieve an 82\% success rate in predicting who would vote for them by tracking voters on over 400 social characteristics using a voter tracking software called Voter Vault \citep{watt-borger}. Similar tactics have allegedly been employed  in the US, Canada and Australia too \citep{howard-kreiss-survey-aus-can-uk-us}.  

Such detailed voter profiles can then perhaps be used for microtargeting and canvassing swing voters through calls, mails or home visits, by distributing cash, alcohol or gifts, or by fear mongering \citep{onselen-electoral-databases,errington-suiting-themselves,howard-kreiss-survey-aus-can-uk-us,special-issue-electoral-fraud-india-pak}; lobbying wealthy voters for party donations \citep{errington-suiting-themselves,howard-kreiss-survey-aus-can-uk-us}; creating hurdles for opposition voters \citep{special-issue-electoral-fraud-india-pak}; maliciously removing culturally distinct names from the electoral roll or raising frivolous objections against their inclusion \cite{kodali-electoralroll-deletion,bhatnagar-electoralroll-deletion,toi-electoralroll-fake-deletion-requests}, etc.  Such profiling often favour the incumbent parties and parties with access to administrative and financial resources and disadvantage the opposition, minor parties and independents \citep{errington-suiting-themselves,howard-kreiss-survey-aus-can-uk-us}. Besides, such excessive focus on data engineering hurts electoral democracy.

The alternative approach of issuing eligibility credentials to eligible voters and checking them during polling suffers from both manipulation by corrupted credential issuers and also disenfranchisement of voters who are unable to safekeep the credential secrets.

Finally, the problem extends beyond just ensuring correct electoral rolls. Even given a correctly prepared electoral roll, auditing that the polling-booth eligibility verification process using the electoral roll was sound --- no ballot stuffing happened (no ineligible voters voted, no eligible voters voted more than once, and no ballots were injected against absentee eligible voters) and no honest eligible voters were denied their right to vote --- is challenging.

In this paper, we examine balancing the conflicting requirements of public auditability and privacy of electoral rolls and polling-booth eligibility verification processes.  To this end, we make the following main contributions:
\begin{enumerate}[leftmargin=*]
	\item We identify the threats of electoral roll manipulation, ballot stuffing, voter denials and privacy violations, and outline the design requirements for a secure electoral roll and eligibility verification protocol (\S\ref{sec:requirements}).
	\item We provide formal security definitions for these requirements (\S\ref{sec:formalisation}). Unlike prior art \citep{juels-coercion-resistance,araujo-coercion-resistance,spycher-coercion-resistance-linear,essex-cobra,clark-selections,araujo-coercion-resistance-fix,ktv-helios}, we do not require any trusted authority for issuing eligibility credentials, or voters to compute zero-knowledge proofs. Our definitions are based on a natural and legally-mandated notion of eligibililty based on factors like age, citizenship, etc.
	\item We propose a secure electoral roll protocol (\S\ref{sec:protocol}) and formally prove its security (\S\ref{sec:analysis}) under certain assumptions (see below). Our protocol is designed for  ``bare-handed'' voters, not requiring them to carry any secret keys or trusted devices, and thus suits large public elections.
	\item We also implement our complete protocol and provide its source code \citep{artifact}. We also provide a program to estimate our protocol's concrete security guarantees in real elections (see \S\ref{sec:pracs}). 
\end{enumerate}

Our protocol relies on two main assumptions. First, we assume the existence of a primary Sybil-resistant identity system that prevents creation of multiple identities of a single person. Such systems are already operational in  national digital identity systems, e.g., in Estonia and India. However, directly using such a primary identity system for electoral processes risks linking voters' electoral data with other databases, thereby compromising voter privacy. To mitigate this, we introduce a mechanism to create a secondary election-specific identity that maintains the Sybil-resistance property but ensures voter anonymity by making it unlinkable to the corresponding primary identity. Our second assumption is that the liveness of a voter in a given electoral environment can be accurately established via facial photographs attested by trusted hardware modules or multiple mutually adversarial local observers. We need this assumption to establish the voter's presence in the polling booth in a bare-handed way. We discuss our assumptions in detail in \S\ref{sec:identity-eligibility}.

\subsection{Related work}
\label{sec:related}

End-to-end verifiable (E2E-V) voting protocols provide cryptographic guarantees to voters that their votes are cast-as-intended, recorded-as-cast and counted-as-recorded. 
However, pollsite E2E-V voting protocols (e.g., \citep{pretavoter,scratchandvote,punchscan}) typically do not address eligibility verifiability and leave it to traditional manual processes.
Eligibility verifiability has been studied in the context of internet voting \citep{juels-coercion-resistance,spycher-coercion-resistance-linear,clark-selections,araujo-coercion-resistance-fix,ktv-helios}, where a trusted registrar issues eligibility credentials to voters who then present them before voting. These works often achieve eligibility verifiability while also maintaining \emph{participation privacy}, i.e., hiding which voters voted and which did not, to protect against forced abstention attacks.

Our work differs from these works in three important aspects. First, these approaches consider the credential-issuing authority as trusted, whereas realistically these authorities may also be corrupted. Second, we aim for ``bare-handed'' protocols where voters are not required to carry any secret keys or trusted devices to prove their eligibility. Non-bare-handed solutions are ill-suited for electorates with low digital literacy or where loss of keys or confiscation of voters' devices by a coercer cannot be ruled out. Third, existing approaches focus only on protecting participation privacy but typically keep the voter list public, whereas our approach goes further to protect voters' identity information from the voter list too, thereby mitigating risks of voter profiling attacks.

\section{Design requirements}
\label{sec:requirements}

\subsection{Current voter registration system}
\label{ssec:current-system}

We begin by describing the currently prevalent voter registration and polling process in a typical public election:
\begin{itemize}[leftmargin=*]
	\item \emph{(Registration).} Voters initiate registration by filing an application, typically containing their name, photograph, residential address, date of birth, age, gender, etc.\footnote{Election officials may also file applications on behalf of some voters.} A registration officer (RO) is responsible for accepting voters' registration information. An eligibility verification officer (EO) assigned for the voter's residential area verifies the voter's eligibility for the election using the submitted data, e.g., by matching the application details with existing public records, verifying voter's age and even making home visits to verify that the voter  indeed resides at the claimed address. If satisfied, the EO approves the voter's entry to be included in the electoral roll. The electoral roll thus created is made public so that anyone can raise objections regarding missing entries, incorrect data or even fake entries and get them resolved till a pre-polling deadline. 
	
	\item \emph{(Polling).} Polling takes place at a polling booth, where the presiding polling officer (PO) checks voters' entries on the electoral roll and marks them off. If a voter's entry is not found or if it is already marked, the PO usually simply denies the voter the chance to vote. Otherwise, the voter is allowed to vote and their (encrypted) vote is recorded against their entry in the electoral roll. 
\end{itemize}

\subsection{Threat model}
\label{ssec:threat-model}

Now we enumerate threats of electoral roll manipulation, both during registration and polling, and threats to voter privacy. Although manipulation threats are sometimes trivial to avoid in public electoral rolls, they become non-trivial to avoid in privacy-preserving electoral rolls.

\subsubsection{Manipulation during registration}
\label{sssec:threats-creation-electoral-rolls}

During registration, the following manipulation threats exist:

\begin{itemize}[leftmargin=*]
	\item \emph{(Denial of registration to eligible voters).} Corrupted ROs and EOs may ignore registration requests of eligible voters and deny registration to them without providing any publicly contestable proof of denial.
	\item \emph{(Data entry errors during registration).} During the voter registration process, misspellings, alternate names, wrong PIN codes or other data entry errors can occur. These errors potentially cause verification failures during polling and disenfranchise the voters.
	\item \emph{(Registration of ineligible or fake voters).} Corrupted ROs and EOs may add fake voters --- dead people, minors, non-citizens, people from a different constituency, or simply ``phantom voters" --- to the electoral roll.
	\item \emph{(Multiple registrations of a single voter).} Corrupted ROs and EOs may register the same voter multiple times, with possibly different aliases, to the electoral roll. Multiple registrations allow ballot stuffing on polling day, either if the same voter is allowed to vote twice or if fake votes are injected by the PO against their duplicated entry.
	\item \emph{(Deletion of previously registered eligible voters).} Corrupted ROs and EOs may delete eligible voters, perhaps belonging to a particular social group, from the electoral roll. This may be done under the guise of purging ineligible voters from the list and without issuing any notice to the affected individuals. 
\end{itemize}

\subsubsection{Manipulation during polling}
\label{sssec:threats-usage-electoral-rolls}

The following threats exist at the polling stage when electoral rolls are used:

\begin{itemize}[leftmargin=*]
	\item \emph{(Denial of voting rights to eligible voters).} During polling, POs may maliciously deny an eligible voter the right to vote or may genuinely misverify the voters' identity or existence in the electoral roll. In privacy-preserving electoral rolls, there may be no easy way to contest the PO's decision.
	\item \emph{(Allowing ineligible voters to vote).} Corrupted POs may allow voters not registered in the electoral roll to cast a vote or inject fake votes against ineligible voters added during registration. 
	\item \emph{(Allowing an eligible voter to vote more than once).} Corrupted POs may allow an eligible voter to vote more than once. They may also inject spurious votes on behalf of eligible voters, potentially after obtaining an access token from when the voter casts their genuine vote for the first time.
	\item \emph{(Injecting votes against an absentee eligible voter).} Corrupted POs may inject votes against eligible voters who did not come to the polling booth to cast their vote.
\end{itemize}

\subsubsection{Privacy threats}
\label{sssec:threats-voter-privacy}

The following privacy threats exist:

\begin{itemize}[leftmargin=*]
	\item \emph{(Voter profiling attacks.)} As mentioned in \S\ref{sec:intro}, voting preference  of a voter may be inferred  with sufficiently high accuracy from information present in the electoral roll, especially when this information is combined with other publicly available information. In addition, political parties may also try to learn other details, e.g., socioeconomic background, lifestyle preferences, shopping habits, etc., to create a detailed profile about each voter. These profiles can then be used to microtarget and manipulate voters.
	\item \emph{(Forced absention attacks.)} Further, if information about whether a voter participated in voting or not gets leaked, it allows coercion of the voter into abstaining from voting. Although participation information necessarily gets leaked to local observers physically present in the polling booth, an electronic voting system should not allow a \emph{remote} observer to systematically obtain and analyse this information.
\end{itemize}

\subsection{High-level requirements}
\label{ssec:high-level-requirements}

Given these threats, we envision that in a secure electoral roll protocol, there must be publicly verifiable audit processes --- which can be invoked by voters, political parties or other interested auditors  --- to verify that no manipulation has taken place, without compromising voters' privacy. 

\subsubsection{Soundness}
\label{sssec:soundness}

The soundness requirement is that it must be computationally hard to make the audits pass with more than a given probability if the electoral roll is manipulated, i.e., if the published encrypted votes do not correspond to the votes cast by eligible voters or if some eligible voters are denied during registration or polling. These audits necessarily involve manual checks on the ground such as cross-checks with local public records, visits to the neighbourhood, etc. Designing an efficient audit process  and while protecting voters' privacy is an interesting technical problem. We discuss these requirements next. 

\subsubsection{Privacy}
\label{sssec:privacy}

A strict privacy requirement would be that the electoral roll should not leak any information about the voters' details other than whether they are eligible. However, we target a weaker version that allows the adversary to obtain identity details of a few randomly selected voters. This is enough to mitigate the  privacy threats posed by electoral rolls. Taken individually, information posted on an electoral roll is not very sensitive (currently all this information is public!); however, it becomes problematic when the entire voter list is made public, enabling systematic building of detailed voter profiles and microtargeting individual vulnerable voters. Leaking a few \emph{randomly} selected voters' identity information would not enable systematic attacks, thereby making them unattractive for the attackers. We also require that even though the adversary might be able to learn eligibility data of a few voters, it should not be able to link the voting information with the unique identity of any voter.

\subsubsection{Efficiency}
\label{sssec:efficiency}

Since the audit process must necessarily involve manual checks, it is infeasible to do this for the millions of entries typically found in an electoral roll. Thus,  to audit an electoral roll consisting of $N$ entries, the audit should require at most $O(\log(N))$ manual checks.

\subsubsection{Bare-handedness}
\label{sssec:bare-handedness}

Finally, for the solution to be usable in public elections, eligibility of a voter must be verifiable using traditional approaches such as identity documents, and should not require voters to carry secret keys or own trusted computing devices. The lack of this requirement leads to disenfranchisement of voters who lose or forget their access cards or their details, have to submit them to a coercer or simply are not digitally literate enough to safekeep such secrets. Also, solutions that rely on a third-party service to store these secrets are undesirable as they pose the risk that the service may cast a vote on behalf of the voters. 

\section{Formalisation}
\label{sec:formalisation}

\subsection{Notation}
	We let $[n]$ denote the set $\{1,\dots,n\}$. We let $\multiset{\cdot}$ denote a \emph{multiset}, i.e., an extension of sets where members may appear multiple times (e.g., $\multiset{a,a,b}$ is a multiset). Given a multiset $M$, $\mult{M}{x}$ denotes the multiplicity of $x$ in $M$ ($\mult{M}{x}=0$ if $x \not\in M$) and $|M|$ denotes $\sum_{x \in M}\mult{M}{x}$. Given two multisets $M_1$ and $M_2$, $M_1 \setminus M_2$ denotes the multiset consisting of each element $x \in M_1$ with multiplicity $\max(\mult{M_1}{x} - \mult{M_2}{x},0)$. We let $\ve{x}$ (boldface) denote a list of values and $x_i$ denote its $i^{\text{th}}$ component. We let $\bot$ denote an error or dummy output. We use $\cdot$ when we do not care about the value taken by a variable.
	
	We let $\mathsf{co}$, $\so{P_1}{\mathsf{so}_1}$, $\dots$, $\so{P_m}{\mathsf{so}_m}$ $\leftarrow$ $\Pi(\mathsf{ci}$, $\so{P_1}{\mathsf{si}_1}$, $\dots$, $\so{P_m}{\mathsf{si}_m})$ denote a multiparty protocol $\Pi$ between parties $P_1,\dots,P_m$ where the common input of each party is $\mathsf{ci}$, $\mathcal{P}_k$'s secret input is $\mathsf{si}_k$, the common output is $\mathsf{co}$ and $\mathcal{P}_k$'s secret output is $\mathsf{so}_k$. In security experiments where some party $P_k$ may be controlled by an adversary $\adv$, we denote this as $P_k^{\adv}$ and drop $P_k$'s secret inputs or outputs.

\subsection{Identity and eligibility}
\label{sec:identity-eligibility}

Our notion of identity is based on individuals'  facial photographs which can be matched either automatically, or even manually. We assume that voters can be publicly identified by matching their live facial photographs with that in the electoral roll, and that faking a live photograph is hard. Our notion of eligibility is based on individuals' extrinsic eligibility data such as address, age, citizenship, and legally mandated eligibility criteria for the election. We assume that it is publicly verifiable -- using standard government issued credentials -- whether a voter identified by (manual) face matching is eligible based on the credentials presented by the voter. We also assume that each potential voter, whether eligible or not, can be assigned a unique and deduplicated identifier based on a Sybil-resistant primary identity system. 

To operationalise these assumptions, we define a set of oracles representing some special physical processes and manual checks. Let $\Voters$ denote the set of all potential voters.
\begin{itemize}[leftmargin=*]
	\item $p \leftarrow \gcapture(V, \env)$: This oracle models the process of capturing a facial photograph $p \in \{0,1\}^*$ of a voter $V \in \Voters$. 
	The argument $\env$ represents the physical environment in which the photograph is taken and models the voter's presence in the given environment and their consent for the photograph to be taken. Physically, this oracle could be implemented by capturing a photograph that shows the voter holding an environment-specific placard. Additionally, to prevent against fake photographs created using modern Deepfake technology, the photograph must be certified by either a trusted hardware module or by multiple mutually adversarial authorities physically present in the environment, e.g., polling agents in a polling booth. We say that $p$ \emph{captures} $V$ if $p=\gcapture(V,\env)$ for some environment $\env$.
	
	\item $0/1 \leftarrow \glive(p,\env)$: This oracle models the process of verifying liveness of a voter in environment $\env$ using photograph $p$. This would entail manually inspecting the photograph and verifying the signatures of the trusted hardware module or the authorities.
	
	\item $0/1 \leftarrow \gmatch(p,p')$: This oracle models the process of verifying whether two different photographs $p$ and $p'$ capture the same voter or not. Physically, this oracle can be implemented by either matching the two photographs completely manually or with the help of an automatic face matching software.
	
	\item $0/1 \leftarrow \gelg(p, \ed)$: This oracle models the process of $a)$ identifying the voter $V \in \Voters$ such that $p = \gcapture(V,\env)$ for some environment $\env$, $b)$ verifying whether $V$ possesses eligibility data $\ed \in \{0,1\}^*$, and $c)$ verifying whether $\ed$ renders $V$ eligible to vote in the election. Physically, this oracle may be implemented by verifying voters' government ID cards, cross-checking with existing public records, home visits, etc. 
	
	\item $0/1 \leftarrow \guid(p, \uid)$: This oracle models the process of verifying whether or not the voter captured by photograph $p$ is assigned an identifier $\uid$ by the trusted primary identity system. Physically, this oracle could be implemented by requiring the identity system to issue identity cards containing photographs and $\uid$s to all potential voters (all citizens), and verifying whether $p$ matches the photograph on an identity card containing the claimed $\uid$. The identity system must have provisions to re-issue identity cards to anyone who forgets their $\uid$. 
\end{itemize}

Our security guarantees hold only if the physical processes represented by the oracles satisfy Assumptions \ref{assumption:CP}-\ref{assumption:CI}.

\begin{assumption}[Consistent photographs] 
	\label{assumption:CP}
	For all photographs $p$ and $p'$ capturing the same voter $V$, it holds that $\gmatch(p, p')=1$. 
\end{assumption}

\begin{assumption}[No identical twins]
	\label{assumption:NT} 
	For all photographs $p,p'$ capturing different voters $V,V' \in \Voters$, it holds that $\gmatch(p, p')=0$.
\end{assumption}

\begin{assumption}[Unforgeable photographs]
	\label{assumption:UP}
	Let $p$ be a photograph capturing a voter $V$ such that $\glive(p,\env)=1$ for some environment $\env$. For any environment $\env' \neq \env$, if $\gcapture(V,\env')$ was never called then it is hard to produce a photograph $p'$ such that $\gmatch(p$, $p')$ $=$ $1$ and $\glive(p',\env')=1$. 
\end{assumption}

\begin{assumption}[Consistent eligibility]
	\label{assumption:CE} 
	For all $p, p',\ed \in \{0,1\}^*$, if $\gmatch(p, p')=1$, it holds that $\gelg(p, \ed) = \gelg(p', \ed)$. Further, if $p$ does not capture any $V \in \Voters$, $\gelg(p,\ed)=0$ for any $\ed$.
\end{assumption}

\begin{assumption}[Deduplicated identifiers]
	\label{assumption:NDI} 
	For any two photographs $p, p'$ and identifiers $\uid, \uid'$ such that $\gmatch(p$, $p')$ $=$ $1$ and $\guid(p$, $\uid)$ $=$ $\guid(p'$, $\uid')$ $=$ $1$, it holds that $\uid = \uid'$.
\end{assumption}

\begin{assumption}[Consistent identifiers]
	\label{assumption:CI} 
	For any two photographs $p, p'$ such that $\gmatch(p$, $p')$ $=$ $1$ and any identifier $\uid$, $\guid(p$, $\uid)$ $=$ $\guid(p'$, $\uid)$.
\end{assumption}

Assumptions \ref{assumption:CP} and \ref{assumption:NT} are essential for reliable public identification of voters via their photographs. Assumption \ref{assumption:UP} models our notion of unforgeability of voters' live photographs (with the necessary certificates) in an environment where the voter was not present and for which they did not give explicit consent. 
Although modern Deepfake technology could potentially create fake voter images, forging signatures of a trusted hardware module or multiple independent physical observers remains challenging enough to make this assumption reasonable in practice. Further, there are additional liveness detection techniques \citep{faceflashing,facecloseup,facelive} that are hard for an AI system to fool.

Assumption \ref{assumption:CE} models consistency of eligiblity over different photographs captured by the same voter and is essential for reliable public verification of eligibility. Assumption \ref{assumption:NDI} models that no voter can be given more than one identifier by the primary identity system and Assumption \ref{assumption:CI} models that the identifiers assigned to individuals should be permanent. Note though that we do not assume that multiple voters are not given the same identifier, as our audits can detect this.

\subsection{Protocol structure}
\label{sec:protocol-structure}

We propose a modular structure for electoral roll protocols that relies on an underlying E2E-V voting protocol $\pietoev$ and clearly separates the eligibility concerns of electoral rolls with the tally correctness concerns of E2E-V voting. Specifically, we assume that voters obtain an encrypted vote $\ev$ from the vote-casting protocol $\pietoev.\cast$ of $\pietoev$, along with a NIZK proof $\rho_{\ev}$ that $\ev$ encrypts a valid vote. Our electoral roll protocol ensures that the list of published $\ev$'s corresponds to those cast by eligible voters, but we rely on $\pietoev$ to guarantee that each $\ev$ correctly encrypts voters' intended votes and that the list of published $\ev$'s is processed correctly to produce the final tally (see Figure \ref{fig:scope}). This modularity allows our protocol to work with any of the existing E2E-V voting protocols such as \citep{scratchandvote,pretavoter,punchscan}. \emph{Note:} We let $\valid(\rho_{\ev}, \ev)=1$ denote the verification of NIZK $\rho_{\ev}$ to verify that $\ev$ is a valid vote. We also assume that $\pietoev$ allows creating encryptions of dummy votes indistinguishable from encryptions of valid votes.

\begin{figure}[h]
	\centering
	\includegraphics[width=\linewidth]{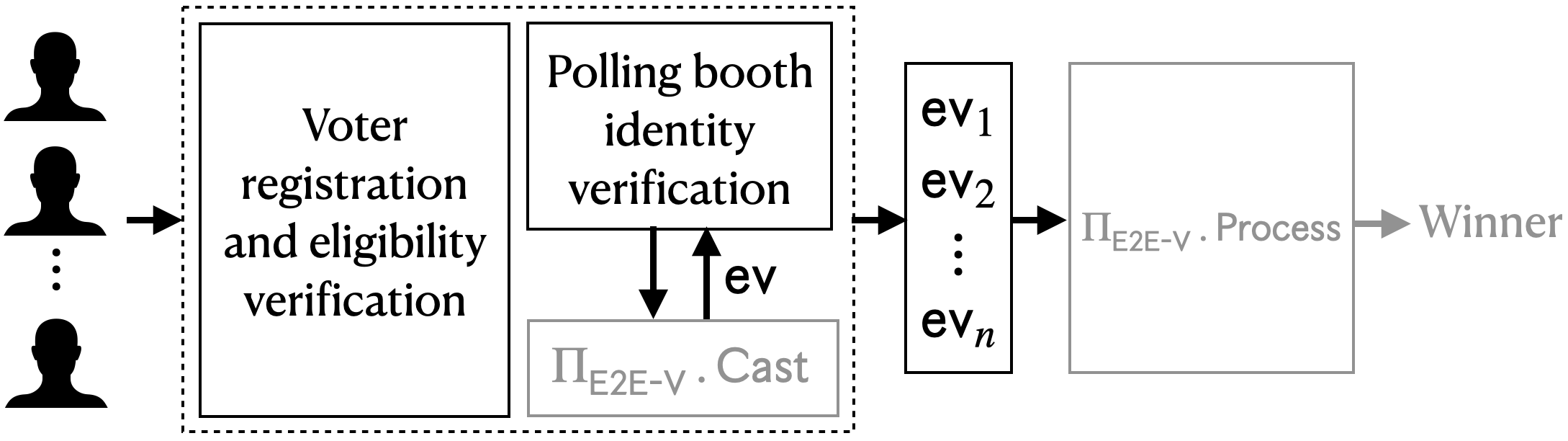}
	\caption{\small{Relationship between our electoral roll protocol and traditional E2E-V voting protocols.}}
	\label{fig:scope}
\end{figure}

In our electoral roll protocol, a set of potential (eligible or not) voters $\Voters$ each acting bare-handedly, and computationally equipped ROs $(R_m)_{m \in [\mu]}$, EOs $(E_j)_{j \in [\jmath]}$, POs $(P_l)_{l \in [\ell]}$, backend servers $(S_k)_{k \in [\kappa]}$, and an auditor $A$ interact with each other via a tuple of sub-protocols $(\setup$, $\register$, $\prepareRoll$, $\cast$, $\indrraudit$, $\indcraudit$, $\univaudit)$, where:
\begin{itemize}[leftmargin=*]
	\item $\pk_S$, $(\so{S_k}{\sk_{S_k}})_{k \in [\kappa]}$ $\leftarrow$ $\setup(1^{\lambda}$, $(\si{S_k}{})_{k \in [\kappa]})$ is a setup protocol among the backend servers. It takes as input the security parameter $\lambda$ in unary and produces a backend public key $\pk_S$ as common output and secret keys for each backend server as their secret output. $\pk_S$ is available as common input to all subsequent sub-protocols.

	\item $\RER, \so{V}{\vid, \rrec, \rvalid}$ $\leftarrow$ $\register(\pk_S$, $\RER$, $\si{V}{\uid, \ed, j}$, $\si{R_m}{}$, $(\si{S_k}{\sk_{S_k}})_{k \in [\kappa]})$ is a voter registration protocol between a voter $V \in \Voters$, an RO $R_m$, and the backend servers. The common input is a list $\RER$ of \emph{encrypted registration requests} made so far (empty initially). $V$'s secret input is its unique identifier $\uid$ obtained from the primary identity system, eligibility data $\ed$ and their residential block number $j$ (this is required to assign the correct EO to verify $V$'s eligibility). The common output is an updated $\RER$ list containing $V$'s request. $V$ obtains a voting identifier $\vid$ to be used during vote casting, a registration receipt $\rrec$, and a bit $\rvalid$ denoting whether $\rrec$ is valid or not.

	\item $\ER,\rho_{\ER}$ $\leftarrow$ $\prepareRoll(\pk_S$, $\RER$, $(\si{S_k}{\sk_{S_k}})_{k \in [\kappa]}$, $(\si{E_j}{})_{j \in [\jmath]})$ is an electoral roll preparation protocol between the backend servers and the EOs, to be executed after the registration phase. The common input is the list $\RER$ of all registration requests. Each $E_j$ approves or dismisses registration requests from voters registering in block $j$. The common output is the prepared electoral roll $\ER$ and a proof of its correctness $\rho_{\ER}$.

	\item $\CI$, $\ve{\ev}$, $\so{V}{\crec,\cvalid}$ $\leftarrow$ $\cast(\pk_S$, $\ER$, $\CI$, $\ve{\ev}$, $\si{V}{\vid, \ev, \rho_{\ev}}$, $\si{P_l}{}$, $(\si{S_k}{\sk_{S_k}})_{k \in [\kappa]})$ is a vote casting protocol between a voter $V \in \Voters$, a PO $P_l$ and the backend servers. The common input is the current state of the electoral roll $\ER$, and lists $\CI$ and $\ve{\ev}$ respectively containing auxiliary casting information and encrypted votes published so far. $V$'s secret input is their voting identity $\vid$, and their encrypted vote $\ev$ and proof of validity $\rho_{\ev}$ obtained from $\pietoev.\cast$. The protocol output is the updated casting information $\CI$ and the updated list of encrypted votes $\ve{\ev}$. $V$'s secret output is their cast receipt $\crec$ and a bit $\cvalid$ denoting whether $\crec$ is valid or not.

	\item $\ver \leftarrow \indrraudit(\pk_S, \RER$, $\si{V}{\rrec, \ed, j}$, $(\si{S_k}{\sk_{S_k}})_{k \in [\kappa]}$, $\si{A}{})$ is an individual registration audit protocol between a voter $V$, auditor $A$ and the backend servers. The common input is the list $\RER$ of registration requests. $V$'s secret input is its registration receipt $\rrec$ and eligibility data $\ed$ and block number $j$ supplied during registration. The output is $A$'s verdict $\ver \in \{0,1\}$ denoting whether $V$ was unfairly denied during registration.
	
	\item $\ver \leftarrow \indcraudit(\pk_S$, $\ER$, $\CI$, $\ve{\ev}$, $\si{V}{\crec, \rrec}$, $(\si{S_k}{\sk_{S_k}})_{k \in [\kappa]}$, $\si{A}{})$ is an individual cast audit protocol among the same parties as $\indrraudit$. The common input is the electoral roll $\ER$, published casting information and encrypted votes $\CI$ and $\ve{\ev}$; $V$'s secret input is its cast receipt $\crec$ and registration receipt $\rrec$. The output is $A$'s verdict $\ver \in \{0,1\}$ denoting whether $V$ was unfairly denied during vote casting. 
	
	\item $\ver \leftarrow \univaudit(\pk_S$, $\alpha$, $\RER$, $\ER$, $\rho_{\ER}$, $\CI$, $\ve{\ev}$, $(\si{S_k}{\sk_{S_k}})_{k \in [\kappa]}$, $\si{A}{})$ is a universal eligibility audit protocol between the backend servers and auditor $A$. The common input is all information published so far and a parameter $\alpha$ controlling the number of verifications. The output is $A$'s verdict $\ver \in \{0,1\}$ denoting whether any large-scale vote stuffing has happened.
\end{itemize}

We assume that all communications between the POs, EOs and the backend servers happen through securely authenticated and encrypted channels. Further, we also assume the existence of a \emph{public bulletin board} \citep{BB0k} for authenticated broadcast among all the parties.

\subsection{Security requirements}

\subsubsection{Completeness} For brevity, we skip a formal description of this requirement. Informally, completeness requires that if all ROs, EOs, POs and backend servers are honest, then audits should pass even if voters are malicious.

\subsubsection{Soundness}

Our soundness requirement (Definition \ref{def:soundness}) captures that even if all ROs, EOs, POs and $S_k$s are corrupted, the audits catch any attempt of large-scale voter denial or vote stuffing. Preventing voter denial entails thwarting attempts to deny honest eligible voters their right to register or vote and attempts to modify or remove votes cast by them. Preventing vote stuffing entails preventing vote casting by ineligible voters, multiple vote casting by a single eligible voter and vote injection against absentee eligible voters. Here, we cannot assume that an eligible voter would act honestly, but we wish to guarantee that they cannot cast more than one vote. 

We model these requirements in experiment $\mathsf{Exp}_{\mathsf{soundness}}$ (Figure \ref{fig:exp-soundness}), parametrised by $n_{\d}$, $\alpha_{\d}$, $f_{\d}$ representing the number of \emph{honest} eligible voters participating in vote casting, number of receipts audited and number of voters denied respectively; and $N_{\s}$, $n_{\s}$, $\alpha_{\s}$, and $f_{\s}$ representing the number of registered eligible (honest or dishonest) voters, number of eligible voters casting their vote, number of ERR/ER entries audited and the number of votes stuffed respectively. The soundness requirement is that the probability of either $f_{\d}$ denials or $f_{\s}$ votes stuffed is less than $\epsilon(n_{\d}$, $\alpha_{\d}$, $f_{\d}$, $N_{\s}$, $n_{\s}$, $\alpha_{\s}$, $f_{\s})$.

The experiment starts by $\adv$ supplying an index set $I_{\h}$ representing the honest voters and initialisation of index sets $I_{\rh}$, $I_{\rdh}$, $I_{\ch}$ and $I_{\cdh}$ representing honest and dishonest voters participating in registration and vote casting respectively. $\EVh$ represents $\ev$'s that should be recorded for the honest voters in the ideal world, modelled as a multiset to allow $\pietoev.\cast$ to issue the same $\ev$ to two different voters (possibly to later clash their votes). We then initialise registration and casting-stage input-outputs for each $(V_i)_{i \in I_{\h}}$.

$\adv$ is then allowed to call oracles $\oreg$ and $\ocast$ to model registration and vote casting of honest voters, where it controls all ROs, EOs, POs and the backend servers, and produce all public information on the bulletin board, with $\EV$ denoting the published encrypted votes (lines 4-5). $\adv$ can also call oracles $\glive,\gmatch,\gelg,\guid$ to verify or query any voter's eligibility or identity information. However, $\adv$ cannot call the $\gcapture$ oracle, which models our central notion of liveness of a voter in an environment using photographs. If $\adv$ gets access to $\gcapture$, it could claim any voter's presence in the polling booth and cast a spurious vote. Nevertheless, $\adv$ obtains a related $\mathsf{OCapture}$ oracle (see below). 

The $\oreg$ and $\ocast$ oracles specify steps to be performed by honest voters to prevent denial for themselves. Voters should register and cast only once, and use the correct $\uid$s and $\vid$s (lines 20 and 28). Also, both during registration and vote casting, voters must perform some manual checks to detect if the receipt issued to them is valid. There must be public processes to prevent the non-issuance of any receipt to the voter. All voters who follow these steps should be registered and allowed to cast, as indicated by sets $I_{\rh}$ and $I_{\ch}$ (lines 21 and 29). Such honest voters' supplied $\ev$'s are added to $\EVh$. $\adv$ wins if $\alpha_{\d}$ random voters' registration receipts and $\alpha_{\d}$ voters' cast receipts are audited successfully by auditor $A$ (lines 6-9) but $|\EVh \setminus \EV| > f_{\d}$, i.e., the number of honest voter $\ev$'s that should have been recorded but were not, either due to denial during registration or casting or due to manipulation of cast votes during publishing, is more than $f_{\d}$.

Prevention against vote stuffing is modelled as follows. First, $\adv$ obtains access to the registration receipt and all information during $\oreg$ and $\ocast$. This models that receipt outputs do not act as eligiblity credentials and do not enable $\adv$ to stuff a vote. Further, with the $\mathsf{OCapture}$ oracle, $\adv$ can capture photographs of honest voters in non-official environments and of dishonest voters in even the casting and registration environments $\envC$ and $\envR$. If a dishonest voter is eligible, though, they are allowed to cast one vote. To capture this, eligible dishonest voters are added to sets $I_{\rdh}$ and $I_{\cdh}$. $\adv$ wins if the universal audit with parameter $\alpha_{\s}$ passes but $|\EV| - |I_{\ch} \cup I_{\cdh}| > f_{\s}$, i.e., the number of published encrypted votes is $f_{\s}$ more than the number of eligible (honest or dishonest) voters participating in vote casting. 

\begin{figure}
	\scalebox{0.8}{
		\begin{tabular}{ @{}rl@{} } 
			1 & $\underline{\mathsf{Exp}_{\mathsf{soundness}}^{\adv,\gset}(1^{\lambda}, n_{\d}, \alpha_{\d}, f_{\d}, N_{\s}, n_{\s}, \alpha_{\s}, f_{\s})}:$ \\
			2 & $I_{\h} \leftarrow \adv(1^{\lambda})$; $I_{\rh}, I_{\rdh}, I_{\ch}, I_{\cdh} \leftarrow \{\}$; $\EVh \leftarrow \multiset{}$ \\
			3 & \textbf{for} $i \in I_{\h}$\textbf{:} $\vc{\uid}{i}^{\r},\vc{\ed}{i}^{\r},\vc{j}{i}^{\r},\vc{\vid}{i}^{\r},\vc{\rrec}{i}^{\r},\vc{\vid}{i}^{\c},\vc{\ev}{i}^{\c},\vc{\crec}{i}^{\c} \leftarrow \bot$ \\
			4 & $\pk_S, \RER, \ER, \rho_{\ER}, \CI, \ve{\ev} \leftarrow \adv^{\mathsf{OReg},\mathsf{OCast},\mathsf{OCapture},\glive,\gmatch,\gelg,\guid}()$  \\
			5 & $\EV := \{\!\!\{\vc{\ev}{i} \mid \vc{\ev}{i} \neq \bot \}\!\!\}$ \\
			6 & \textbf{for} $\alpha_{\d}$ randomly selected $i \in I_{\rh}$\textbf{:} \\
			7 & \quad $\ver_{\rrec_i^{\r}} \leftarrow \indrraudit(\pk_S, \RER, \si{V_i}{\rrec_i^{\r},\ed_i^{\r},j_i^{\r}}, (S_k^{\adv})_{k \in [\kappa]}, \si{A}{})$ \\ 
			8 & \textbf{for} $\alpha_{\d}$ randomly selected $i \in I_{\ch}$\textbf{:} \\
			9 & \quad $\ver_{\crec_i^{\c}} \leftarrow \indcraudit(\pk_S, \ER, \CI, \ve{\ev}, \si{V_i}{\crec_i^{\c}, \rrec_i^{\r}}, (S_k^{\adv})_{k \in [\kappa]}, \si{A}{})$ \\
			10 & $\mathsf{ver}_{\mathsf{UA}} \leftarrow \univaudit(\pk_S, \alpha_{\s}, \RER, \ER, \rho_{\ER}, \CI, \ve{\ev}, (S_k^{\adv})_{k \in [\kappa]}, \si{A}{})$ \\
			11 & $\ver = 1$ iff $(\forall i: \mathsf{ver}_{\rrec_i^{\r}} = 1) \wedge (\forall i: \mathsf{ver}_{\crec_i^{\c}} = 1) \wedge \mathsf{ver}_{\UA} = 1$ \\
			12 & \textbf{return} $1$ \textbf{if} $\ver = 1$ \textbf{and} \\
			13 & \quad ($|I_{\ch}|=n_{\d}$ \textbf{and} $|\EVh \setminus \EV| > f_{\d}$) \textbf{or} \\
			14 & \quad ($|I_{\rh}|+|I_{\rdh}|=N_{\s}$ \textbf{and} $|I_{\ch}|+|I_{\cdh}|=n_{\s}$ \textbf{and} $|\EV| - |I_{\ch} \cup I_{\cdh}| > f_{\s}$) \\
			15 & \\
			16 & $\mathsf{OReg}(i \in I_{\h}, \uid, \ed, m, j, \RER):$ \\
			17 & \quad $\RER, \so{V_i}{\vid, \rrec, \rvalid} \leftarrow \register(\pk_S, \RER, \si{V_i}{\uid, \mathsf{ed}, j}, R_m^{\adv},$ \\
			18 & \quad \quad $(S_k^{\adv})_{k \in [\kappa]})$ \\
			19 & \quad $p \leftarrow \gcapture(V_i, \envR)$ \\
			20 & \quad \textbf{assert} $i \not\in I_{\rh}$ \textbf{and} $\guid(p, \uid)=1$ \textbf{and} $\rvalid = 1$ \\
			21 & \quad $I_{\rh} := I_{\rh} \cup \{ i \}; \vc{\uid}{i}^{\r} := \uid$; $\vc{\ed}{i}^{\r} := \ed$; $\vc{j}{i}^{\r} := j$; $\vc{\vid}{i}^{\r} := \vid$; $\vc{\rrec}{i}^{\r} := \rrec$ \\
			22 & \quad \textbf{return $\RER, \vid, \rrec, \rvalid$} \\
			23 & \\
			24 & $\mathsf{OCast}(i \in I_{\h}, \vid, \ev, \rho_{\ev}, l, \ER, \ve{\ev}):$ \\
			25 & \quad $\CI$, $\ve{\ev}$, $\so{V_i}{\crec, \cvalid} \leftarrow \cast(\pk_S, \ER, \CI, \ve{\ev}, \si{V_i}{\vid, \ev, \rho_{\ev}}, P_{l}^{\adv},$ \\
			26 & \quad \quad $(S_k^{\adv})_{k \in [\kappa]})$ \\
			27 & \quad $p \leftarrow \gcapture(V_i, \envC)$ \\
			28 & \quad \textbf{assert} $i \in I_{\rh} \wedge \vid = \vid_i^{\r} \wedge \gelg(p,\ed_i^{\r}) = 1 \wedge i \not\in I_{\ch} \wedge \cvalid = 1$ \\
			29 & \quad $\EVh := \EVh \cup \{\!\!\{ \ev \}\!\!\}$; $I_{\ch} := I_{\ch} \cup \{ i \}$ \\
			30 & \quad $\vc{\vid}{i}^{\c} := \vid$; $\vc{\ev}{i}^{\c} := \ev$; $\vc{\crec}{i}^{\c} := \crec$ \\
			31 & \quad \textbf{return $\CI, \ev, \crec, \cvalid$} \\
			32 & \\
			33 & $\mathsf{OCapture}(i, \env):$ \\
			34 & \quad \textbf{if} $i \in I_{\h}$\textbf{:} \\
			35 & \quad \quad \textbf{assert} $\env \not\in \{\envR, \envC\}$; $p \leftarrow \gcapture(V_i, \env)$ \\
			36 & \quad \textbf{else}\textbf{:} \\
			37 & \quad \quad $p \leftarrow \gcapture(V_i, \env)$ \\
			38 & \quad \quad \textbf{if} $(\exists \ed: \gelg(p,\ed)=1)$\textbf{:} \\
			39 & \quad \quad \quad \textbf{if} $\env = \envR$\textbf{:} $I_{\rdh} := I_{\rdh} \cup \{ i \} $ \\
			40 & \quad \quad \quad \textbf{if} $\env = \envC$\textbf{:} $I_{\cdh} := I_{\cdh} \cup \{ i \} $ \\
			41 & \quad \quad \textbf{return} $p$ 
		\end{tabular}
	}
	\caption{Soundness experiment}
	\label{fig:exp-soundness}
\end{figure}  

\begin{definition}[Soundness]
	\label{def:soundness}
	Given a function $\epsilon : \Nat^7 \rightarrow [0,1]$, an electoral roll protocol $(\setup$, $\register$, $\prepareRoll$, $\cast$, $\indrraudit$, $\indcraudit$, $\univaudit)$ is called \emph{$\epsilon$-sound} if for all parameters $\lambda$, $n_{\d}$, $\alpha_{\d}$, $f_{\d}$, $N_{\s}$, $n_{\s}$, $\alpha_{\s}$, $f_{\s}$ $\in$ $\Nat$, all oracles $\gset = \{\gcapture$, $\glive$, $\gmatch$, $\gelg$, $\guid\}$ satisfying Assumptions \ref{assumption:CP}-\ref{assumption:CI}, and all PPT adversaries $\adv$, the following holds, where $\mathsf{Exp}_{\mathsf{soundness}}^{\adv,\gset}$ is as defined in Figure \ref{fig:exp-soundness}:
	\begin{align*}
		\begin{split}
			& \Pr[\mathsf{Exp}_{\mathsf{soundness}}^{\adv,\gset}(1^{\lambda}, n_{\d}, \alpha_{\d}, f_{\d}, N_{\s}, n_{\s}, \alpha_{\s}, f_{\s}) = 1] \leq \\
			& \quad \epsilon(n_{\d}, \alpha_{\d}, f_{\d}, N_{\s}, n_{\s}, \alpha_{\s}, f_{\s}),
		\end{split}
	\end{align*}
\end{definition}

\subsubsection{Privacy}
\label{sec:privacy}

Definition \ref{def:privacy} models our privacy definition designed to protect against the voter profiling, linking and forced abstention attacks described in \S\ref{ssec:threat-model}. We assume that POs and EOs are trusted for this definition because they necessarily learn voters' identity and eligibility information to verify it. Similarly, ROs are trusted because they directly obtain identity information from bare-handed voters. These assumptions are further justified as attempts by POs, EOs or ROs to deviate from the protocol risk failing the audits and potential punitive actions. We also assume that at least one backend server is honest. This assumption is standard for secrecy in voting protocols.

Our requirements essentially capture that a protocol where a small $\alpha$ out of $n$ random ERR/ER entries and $\alpha$ random receipts are audited for providing soundness should only provide a small $\delta(n,\alpha)$ advantage to the adversary in learning about voters' eligibility data, participation information and linkage between $\uid$ and $\vid$. This is a reasonable guarantee because leaking a few random voters' identity, eligibility and participation information is not very sensitive when taken individually and leads to profiling and microtargeting attacks only when it is available in bulk for all the voters.

We also wish that privacy concerns do not deter voters from verifying their receipts. However, if more than $\alpha$ voters verify their receipts, the adversary necessarily learns eligibility and participation information for those voters. Nevertheless, the third requirement ensures that even for these voters, the adversary does not get any additional advantage in linking the $\uid$ and $\vid$ information.

\begin{definition}[Privacy] 
	\label{def:privacy}
	Given a function $\delta : \Nat \times \Nat \rightarrow [0,1]$, an electoral roll protocol $(\setup$, $\register$, $\prepareRoll$, $\cast$, $\indrraudit$, $\indcraudit$, $\univaudit)$ is called \emph{$\delta$-private} if for all security parameters $\lambda \in \Nat$, all oracles $\gset = \{\gcapture$, $\glive$, $\gmatch$, $\gelg$, $\guid\}$ satisfying Assumptions \ref{assumption:CP}-\ref{assumption:CI}, and all PPT adversaries $\adv$ controlling auditor $A$, all-but-one backend servers $(S_k)_{k \neq k^*}$, and all voters unless specified below (but not any $(R_m)_{m \in [\mu]}$, $(P_l)_{l \in [\ell]}$, or $(E_j)_{j \in [\jmath]}$), if $\univaudit$ is called with parameter $\alpha$, and $|\ve{\ev}| \geq n$, then:
	\begin{enumerate}[leftmargin=*]
		\item \textbf{Prevention of voter profiling.} If at least $n$ honest voters participate in $\register$ and $\cast$ and at most $\alpha$ random honest voters participate in $\indrraudit$ and $\indcraudit$ each, then for any honest voter $V_i$ participating in $\register$ and any $(\uid, \ed_0, \ed_1, j)$, $\adv$'s advantage in distinguishing between a world where $V_{i}$ registers with $(\uid, \ed_0, j)$ and a world where $V_i$ registers with $(\uid, \ed_1, j)$ is at most negligibly more than $\delta(n,\alpha)$.
		\item \textbf{Prevention of forced abstention attacks.} If at least $n$ honest voters participate in $\register$ and $\cast$ and at most $\alpha$ random honest voters participate in $\indrraudit$ and $\indcraudit$ each, then for any two honest voters $V_{i_0}$ and $V_{i_1}$ participating in $\register$, $\adv$'s advantage in distinguishing between a world where $V_{i_0}$ votes and $V_{i_1}$ abstains and a world where $V_{i_1}$ votes and $V_{i_0}$ abstains is at most negligibly more than $\delta(n,\alpha)$.
		\item \textbf{Unlinkability of identity and voting information.} For any honest voter $V_{i}$ participating in $\register$ and any $(\uid_0$, $\uid_1$, $\ed_0$, $\ed_1$, $j_0$, $j_1)$, $\adv$'s advantage in distinguishing between a world where $V_{i}$ registers with $(\uid_{\tau}$, $\ed_{\tau}$, $j_{\tau})$ and a world where $V_{i}$ registers with $(\uid_{\tau}$, $\ed_{1-\tau}$, $j_{1-\tau})$, where $\tau \in \{0,1\}$ is randomly chosen, is at most negligibly more than $\delta(n,\alpha)$. 
	\end{enumerate}
\end{definition}

\section{Preliminaries}
\label{sec:preliminaries}

\emph{Threshold encryption.} An IND-CPA secure $(\kappa,\kappa)$-threshold encryption scheme $\mathsf{E}^{\text{th}}:=(\mathsf{Keygen}$, $\mathsf{Enc}$, $\thsd)$ with message and ciphertext spaces $\mathsf{M}(\ths)$ and $\mathsf{C}(\ths)$, where:
\begin{itemize}[leftmargin=*]
	\item $\pk_S, (\so{S_k}{\sk_{S_k}})_{k \in [\kappa]} \leftarrow \thsg(1^{\lambda}, (\si{S_k}{})_{k \in [\kappa]})$ is a protocol that generates a common public key $\pk_S$ and a secret key $\sk_{S_k}$ for each $S_k$.
	\item $c \leftarrow \thse(\pk_S, m)$ is an algorithm to encrypt a message $m$ under public key $\pk_S$ to produce a ciphertext $c$.
	\item $(m, \rho)$ $\leftarrow$ $\thsd(c$, $(\si{S_k}{\sk_{S_k}})_{k \in [\kappa]})$ is a threshold decryption protocol where each $S_k$ partially (and verifiably) decrypts $c$ and combines their partial decryptions to obtain the final decryption $m$ and a proof of correctness $\rho$. We let $\so{P}{m, \rho} \leftarrow \thsd(\pk_S, c$, $(\si{S_k}{\sk_{S_k}})_{k \in [\kappa]}$, $\si{P}{})$ denote a variant of this where each $S_k$ sends their partial decryptions (along with proofs of correctness of shares) directly to another party $P$ via a secret channel. $P$ thus obtains the decryption $m$ and proof $\rho$ but $S_k$'s do not. 
\end{itemize}
IND-CPA security of such schemes is defined analogously to that of standard public-key encryption schemes, under the condition that the adversary does not corrupt all $\kappa$ decryptors.

We require $\ths$ to also support a \emph{re-encryption} operation $c' \leftarrow \thsre(\pk_S, c)$ that takes as input $\pk_S$, $c = \thse(\pk_S, m)$ and produces $c'$ indistinguishable from a fresh encryption of $m$. We also require that $\ths$ supports $a)$ a NIZK proof $\mathsf{NIZK}_{\mathsf{enc}}(\pk_S, m, c)$ for proving knowledge of a secret message $m$ encrypted to a public ciphertext $c$ under public key $\pk_S$; and $b)$ a NIZK proof $\mathsf{NIZK}_{\mathsf{eq}}(\pk_S, c_1, c_2)$ for proving that ciphertexts $c_1$ and $c_2$ encrypt the same secret message under public key $\pk_S$. Specifically, we choose the threshold El Gamal encryption scheme \citep{threshold-cryptography}, which supports all these requirements efficiently with standard techniques.

\emph{Shuffles and proofs of shuffle.} We let $\ve{c}'$, $\rho_{\mathsf{shuf}}$ $\leftarrow$ $\mathsf{Shuffle}(\pk_S$, $\ve{c}$, $(\si{S_k}{})_{k \in {\kappa}})$ denote a shuffle protocol. Each $S_k$ in $(S_1,\dots,S_{\kappa})$ successively takes as input a list $\ve{c}$ of $\ths$ ciphertexts under $\pk_S$, creates a re-encrypted and permuted list $\ve{c'}$, where $c'_i:=\thsre(\pk_S, c_{\pi_k(i)} )$ under a fresh permutation $\pi_k$, and forwards $\ve{c}:=\ve{c}'$, along with a proof of correctness $\rho_{\mathsf{shuf},k}$, to $S_{k+1}$. The output of the last party $S_{\kappa}$ is the output ciphertext list $\ve{c}'$. The output proof $\rho_{\mathsf{shuf}}$ is simply the set $\{ \rho_{\mathsf{shuf},k} \}_{k \in [\kappa]}$. Proving that multiple ciphertext lists $\ve{c}_1,\dots,\ve{c_m}$ are obtained correctly by shuffling ciphertext lists $\ve{c}_1',\dots,\ve{c}_m'$  respectively by \emph{the same secret permutation $\pi$} shared across $(S_k)_{k \in {\kappa}}$ can be done using the \emph{commitment-consistent proof-of-shuffle techniques} of \citep{commitment-consistent-proof-shuffle,terelius-restricted-shuffles}. 

\emph{Additional conventions.} We club multiple invocations of a protocol into the same call wherever possible, e.g., $(c_1$, $c_2)$ $\leftarrow$ $\thse(\pk_S$, $(m_1$, $m_2))$ abbreviates $c_1$ $\leftarrow$ $\thse(\pk_S$, $m_1)$; $c_2$ $\leftarrow$ $\thse(\pk_S$, $m_2)$, and $\ve{c}$ $\leftarrow$ $\thse(\pk_S$, $\ve{m})$ abbreviates encryption of each plaintext in list $\ve{m}$ to obtain a list of ciphertexts $\ve{c}$. We let $\oplus$ denote the bitwise XOR operation. We let $H_{\alpha}$ denote a hash function that takes as input a set $S$ and outputs a random subset $S' \subseteq S$ of size $\alpha$. Practically, this can be implemented by taking the output of a standard hash function as the random tape of a randomised algorithm that performs such sampling.

\section{Our protocol}
\label{sec:protocol}

\begin{figure}[t]
    \centering
    \scalebox{0.7}{
        \begin{tabular}{@{\hskip 0.0cm}l@{\hskip 0.02cm} c@{\hskip 0.02cm} l@{\hskip 0.02cm} c@{\hskip 0.0cm}}
			\begin{tabular}[t]{|c|}
                \hline
                \cellcolor{gray!20}{\small $\ERR$} \\
                \hline
                $\uid_1, c_{\vid_{1}}, c_{j_{1}}, c_{\ed_{1}}, c_{\bidR_1}$ \\
                \hline
                $\uid_2, c_{\vid_{2}}, c_{j_{2}}, c_{\ed_{2}}, c_{\bidR_2}$ \\
                \hline
                $\uid_3, c_{\vid_{3}}, c_{j_{3}}, c_{\ed_{3}}, c_{\bidR_3}$ \\
                \hline
                $\uid_4, c_{\vid_{4}}, c_{j_{4}}, c_{\ed_{4}}, c_{\bidR_4}$ \\
                \hline
                $\uid_5, c_{\vid_{5}}, c_{j_{5}}, c_{\ed_{5}}, c_{\bidR_5}$ \\
                \hline
            \end{tabular} & 
            \begin{tabular}[t]{c|c|c}
				\multicolumn{3}{c}{} \\
                \multicolumn{3}{c}{Mixnet} \\
                \cline{2-2}
                $\rightarrow$ & \multirow{5}{*}{$\pi$} & $\rightarrow$\\
                $\vdots$ & & $\vdots$ \\
                $\vdots$ & & $\vdots$\\
                $\rightarrow$ & & $\rightarrow$\\
                \cline{2-2}
            \end{tabular} &
            \begin{tabular}[t]{|c@{\hskip 0.0cm}|c@{\hskip 0.0cm}|c|}
                \hline
                {\cellcolor{gray!20}{\small $\ER$}} & {\cellcolor{gray!20}{\small $\CI$}} & {\cellcolor{gray!20}{\small $\ve{\ev}$}} \\
                \hline
                $\vid_1',j_1',c'_{\ed_1},c'_{\bidR_1},\rresp_1,l_1$ & $c_{\bidC_1},c_{\rho_{\ev_1}}$ & $\ev_1$ \\
                \hline
                $\vid_2',j_2',c'_{\ed_2},c'_{\bidR_2},\rresp_2,l_2$ & $c_{\bidC_2},c_{\rho_{\ev_2}}$ & $\ev_2$ \\
                \hline
                $\vid_3',j_3',c'_{\ed_3},c'_{\bidR_3},\rresp_3,l_3$ & $c_{\bidC_3},c_{\rho_{\ev_3}}$ & $\ev_3$ \\
                \hline
                $\vid_4',j_4',c'_{\ed_4},c'_{\bidR_4},\rresp_4,l_4$ & $c_{\bidC_4},c_{\rho_{\ev_4}}$ & $\ev_4$ \\
                \hline
                $\vid_5',j_5',c'_{\ed_5},c'_{\bidR_5},\rresp_5,l_5$ & $c_{\bidC_5},c_{\rho_{\ev_5}}$ & $\ev_5$ \\
				\hline
            \end{tabular} 
        \end{tabular}
    }
    \caption{Outputs published in our protocol.}
    \label{fig:protocol-overview}
\end{figure}

\subsection{High-level overview}

We first give a high-level overview of our protocol. During setup, the backend servers generate a public key $\pk_S$ under a threshold encryption scheme $\ths$ such that the secret key is shared among the servers. All our encryptions are under $\pk_S$.

For registration, each voter interacts with an RO $R_m$ in a registration booth $m$, supplying their unique identifier $\uid$ obtained from the primary identity system, their eligibility data $\ed$ for the election and the residential block $j$ under which they wish to register. $R_m$ issues an election-specific unlinkable identifier $\vid$ to the voter in a registration receipt and uploads encryptions $c_{\vid}$, $c_j$ and $c_\ed$ of $\vid$, $j$ and $\ed$ respectively to a public bulletin board. $R_m$ also captures a live photograph $\bidR$ of the voter in environment $\envR$ and uploads its encryption $c_{\bidR}$ too as authorisation for $c_{\vid}$, $c_j$ and $c_\ed$. These encryptions together represent encrypted registration requests accepted by $R_m$ and are uploaded as list $\ERR$ on a public bulletin board (see Figure \ref{fig:protocol-overview}), whereas the voters are issued registration acceptance receipts. Voters who, as per $R_m$, were trying to register twice\footnote{Note that we present a simplified static model here, where voters supply their eligibililty data only once. We briefly discuss a dynamic model in \S\ref{sec:ai-dynamic}.} or register using a wrong $\uid$ are given rejection receipts. A public process ensures that no voter is denied without issuing a receipt.

The $\ERR$ entries are processed by backend servers $(S_k)_{k \in [\kappa]}$ in a mixnet fashion, provably shuffling lists $\ve{c_{\vid}}$, $\ve{c_{j}}$, $\ve{c_{\ed}}$ and $\ve{c_{\bidR}}$ under a shared secret permutation $\pi$ to obtain lists $\ve{c'_{\vid}}$, $\ve{c'_{j}}$, $\ve{c'_{\ed}}$ and $\ve{c'_{\bidR}}$ respectively. Subsequently, $\ve{c'_{\vid}}$ and $\ve{c'_{j}}$ are provably threshold-decrypted to produce lists $\ve{\vid'}$ and $\ve{j'}$ respectively. Each EO $E_j$ obtains decryptions of $c_{\ed_i}'$ and $c_{\bidR_i}'$ for each $\vid'_i$ entry marked with block identifier $j$, uses this data to verify the voter's eligibility and uploads their decision $\rresp_i \in \{0,1\}$ against $\vid'_i$. $E_j$ also uploads the booth identifier $l$ assigned to the voter and communicates it to the voter. This list, containing voters' eligibility data against their $\vid$ identities unlinkable to their $\uid$ identities, along with EOs' approvals, is called $\ER$ or the \emph{electoral roll}. At this point, a random audit of the electoral roll could be done to verify EOs' approvals, and voters could request audit of their registration receipts and file complains. 

During vote casting, voters accepted during registration visit the polling booth $l$ assigned to them, carrying the $\vid$ in their registration receipt. The PO $P_l$ obtains the voter's registered photograph $\bidR'$ decrypted from $c'_{\bidR}$ uploaded against $\vid$. The backend servers allow $P_l$ to obtain these decryptions only for $\vid$ entries marked with booth $l$. $P_l$ allows the voter to cast their vote if and only if $\bidR'$ matches the live voter's photograph, EO's decision $\rresp_i$ for $\vid$ is $1$ and the voter has not voted already. If allowed, the voter obtains an encrypted vote $\ev$ and a proof of validity $\rho_{\ev}$ by following the underlying $\pietoev.\cast$ protocol. $P_l$ scans both $\ev$ and $\rho_{\ev}$ and uploads $\ev$ and an encryption $c_{\rho_{\ev}}$ of $\rho_{\ev}$ as casting information $\CI$ against $\vid$. $P_l$ also uploads an encryption $c_{\bidC}$ of a live photograph $\bidC$ of the voter in environment $\envC$ as authorisation for $\ev$ and $c_{\rho_{\ev}}$. 

Note that $\rho_{\ev}$ is encrypted to hide who voted. For the same reason, $P_l$ also uploads dummy encrypted votes $\ev_{\bot}$ and encryptions $c_{\rho_{\ev}}$ of dummy proofs for all $\vid$s in booth $l$ corresponding to absentee voters. Further, if $\pietoev.\process$ can remove dummy votes anonymously (say, in a mixnet-based backend \citep{pretavoter,mixnet-sok}) then $\ve{\ev}$ are directly fed to it, otherwise (say, in a homomorphic tallying-based backend \citep{scratchandvote,starvote}) $(S_k)_{k \in [\kappa]}$ first anonymise $(\ve{\ev},\ve{c_{\rho_{\ev}}})$ through a separate mixnet, obtaining $\ev$ and decryptions of $\ve{c_{\rho_{\ev}}}$ at the end of the mixnet, and then feed valid $\ev$'s from this output to $\pietoev.\process$. 

As with registration, accepted voters are given cast acceptance receipts and voters who, as per $P_l$, were deemed ineligible by the EO or trying to vote twice or against someone else's $\vid$ are issued rejection receipts.

Vote stuffing is prevented by decrypting a few randomly selected $\ERR$ and $\ER$ entries and verifying the live photographs and eligibility of voters. Multiple entries in $\ER$ by the same voter are prevented by a zero-knowledge proof of deduplicated $\vid$s given deduplicated $\uid$s from Assumption \ref{assumption:NDI}. Eligibility of voters is established by the $\ER$ audit checking that voters' registered photographs and eligibility data indeed makes them eligible as per $\gelg$. Ballot stuffing against absentee eligible voters is established by the $\ER$ audit checking if the decrypted casting stage photographs match the corresponding registered photographs. 

Prevention of voter denials requires overcoming many conflicting requirements. First, bare-handed voters cannot easily detect and contest attempts to deny them by computationally equipped adversaries. Thus, part of the verification must be done during a post-facto receipt audit by a computationally equipped auditor, but this should not compromise voters' privacy to the auditor. However, proofs about our identity and eligibility notions inherently rely on manual checks, which necessitate decryption of encrypted information and cannot be given via cryptographic zero-knowledge proofs. Further, there is the threat of malicious voters trying to register or cast against other voters' entries, knowing they will be denied, but hoping to learn information about them. We carefully design our $\register$ and $\cast$ protocols to tackle these challenges.

Our soundness guarantee comes from the fact that each vote stuffing attempt leads to a distinct ``bad'' $\ERR/\ER$ entry that would be caught if audited and that each voter denial attempt leads to a bad registration or cast receipt. Thus, verification of a few random $\ERR/\ER$ entries or a few random receipts detects any large-scale vote stuffing or voter denial attacks. Our privacy guarantee mainly comes from the fact that only a few random $\ERR/\ER$ entries are opened.

Figures \ref{fig:protocol-register}, \ref{fig:protocol-prepareER} and \ref{fig:protocol-cast} show the registration, $\ER$ preparation and vote casting phases of our protocol in detail. In \S\ref{sec:registration}, \ref{sec:prepareER} and \ref{sec:casting}, we highlight how vote stuffing, voter denial and privacy attacks are prevented in each of these phases (see \S\ref{sec:analysis} for complete formal proofs). In \S\ref{sec:ai}, we discuss some additional practical issues.

\subsection{Registration (Figure \ref{fig:protocol-register})}
\label{sec:registration}

\begin{figure}[t]
	\scalebox{0.8}{
		\begin{tabular}[t]{@{}l@{}}
			\hline
			$\setup(1^{\lambda}, (\si{S_k}{})_{k \in [\kappa]})$: \\
			\hline
			\textbf{output} $\pk_S, (\so{S_k}{\sk_{S_k}})_{k \in [\kappa]} \leftarrow \thsg(1^{\lambda}, (\si{S_k}{})_{k \in [\kappa]})$ \\
			\hline
			\\
			\hline
			$\register(\pk_S, \ERR = (\ve{\uid}, \cdot, \cdot, \cdot, \ve{c_{\bidR}}), \si{V}{\uid, \ed, j}, \si{R_m}{}, (\si{S_k}{\sk_{S_k}})_{k \in [\kappa]})$: \\
			\hline
			$V \rightarrow R_m$: \enskip $\uid, \ed, j$  \\
			$V,R_m$: \enskip $\bidpreR \leftarrow \gcapture(V, \envpreR)$  \\

			$R_m$: \enskip \textbf{if} $\guid(\bidpreR, \uid) \neq 1$\textbf{:} \quad \enskip \mycomment{Rejected - unrecognised or incorrect $\uid$}\\
			\quad \quad \enskip \quad $R_m \rightarrow V$: \enskip $\rrec:=(\reject, \uid, \bidpreR)$ \\
			\quad \quad \enskip \quad $V$: \enskip $\rvalid:=1$ iff $\rrec=(\reject, \uid, \bidpreR)$; $\vid:=\bot$ \\
			
			$R_m$: \enskip \textbf{else if} $\exists i: \uid_i=\uid$\textbf{:} \quad \quad \quad \enskip \ \mycomment{Rejected - duplicate registration}\\
			\quad \quad \enskip \quad $\so{R_m}{\bidR_i, \cdot} \leftarrow \thsd(\pk_S, c_{\bidR_i}, (\si{S_k}{\sk_{S_k}})_{k \in [\kappa]}, \si{R_m}{})$ \\
			\quad \quad \enskip \quad \textbf{assert} $\gmatch(\bidR_i, \bidpreR)=1$ \\  
			\quad \quad \enskip \quad $R_m \rightarrow V$: $\rrec:=(\alreadyreg, \bidpreR, \bidR_i)$ \\
			\quad \quad \enskip \quad $V$: \enskip $\rvalid:=1$ iff $\rrec=(\alreadyreg, \bidpreR, \cdot)$; $\vid:=\bot$ \\
			
			$R_m$: \enskip \textbf{else:} \quad \quad \quad \quad \quad \quad \quad \quad \quad \quad \quad \quad \quad \quad \mycomment{Accepted for registration} \\
			\quad \quad \enskip \quad $V,R_m$: \enskip $\bidR \leftarrow \gcapture(V, \envR)$ \\
			\quad \quad \enskip \quad $\vid \xleftarrow{\$} \mathsf{M}(\ths)$ \\
			\quad \quad \enskip \quad $(c_{\vid}, c_{j}, c_{\ed}, c_{\bidR}) \leftarrow \thse(\pk_S, (\vid, j, \ed, \bidR))$ \\
			\quad \quad \enskip \quad $\ERR.\mathsf{append}((\uid, c_{\vid}, c_j, c_{\ed}, c_{\bidR}))$ \\
			\quad \quad \enskip \quad $(c_{\vid}^*,c_{j}^*,c_{\ed}^*) \leftarrow \thse(\pk_S, (\vid, j, \ed))$ \\
			\quad \quad \enskip \quad $\rho_{\mathsf{enc},1} \leftarrow \nizk_{\mathsf{enc}}(\pk_S, \bidR, c_{\bidR})$ \\
			\quad \quad \enskip \quad $\rho_{\mathsf{enc},2} \leftarrow \nizk_{\mathsf{enc}}(\pk_S, (\vid, j, \ed), (c_{\vid}^*, c_{j}^*, c_{\ed}^*))$ \\
			\quad \quad \enskip \quad $\rho_{\mathsf{eq}} \leftarrow \nizk_{\mathsf{eq}}(\pk_S, ((c_{\vid}, c_{\vid}^*), (c_j,c_j^*), (c_{\ed},c_{\ed}^*)))$ \\
			\quad \quad \enskip \quad $\rho_{\mathsf{eq},1} \xleftarrow{\$} \{0,1\}^{|\rho_{\mathsf{eq}}|}; \rho_{\mathsf{eq},2} \leftarrow \rho_{\mathsf{eq},1} \oplus \rho_{\mathsf{eq}}$ \\
			\quad \quad \enskip \quad $\rrec_{11}:=(\uid, \bidR, \rho_{\mathsf{enc},1}), \rrec_{12}:=(c_{\vid}, c_j, c_{\ed}, c_{\bidR}, \rho_{\mathsf{eq},1})$ \\
			\quad \quad \enskip \quad $\rrec_{21}:=(\vid, j, \ed, \rho_{\mathsf{enc},2}), \rrec_{22}:=(c_{\vid}^*, c_j^*, c_{\ed}^*, \rho_{\mathsf{eq},2})$ \\
			\quad \quad \enskip \quad $R_m \rightarrow V$: \enskip $\rrec:= (\accept, \rrec_{11}, \rrec_{12}, \rrec_{21}, \rrec_{22})$ \\
			\quad \quad \enskip \quad $V$: \enskip $\rvalid := 1$ iff $\rrec = (\accept, (\uid, \bidR, \cdot), \cdot, (\vid, j, \ed, \cdot), \cdot)$ \\
			\textbf{output} $\ERR = (\ve{\uid}, \ve{c_{\vid}}, \ve{c_{j}}, \ve{c_{\ed}}, \ve{c_{\bidR}}), \so{V}{\vid, \rrec, \rvalid}$ \\
			\hline

			\\

			\hline
			$\indrraudit(\mathsf{pk}_S, \ERR, \si{V}{\rrec, \ed, j}, (\si{S_k}{\mathsf{sk}_{S_k}})_{k \in [\kappa]}, \si{A}{})$: \\
			\hline
			$V$: \enskip \textbf{if} $\rrec = (\reject, \tilde{\uid}, \tilde{\bidpreR})$\textbf{:} \\
			\quad \quad \quad $V \rightarrow A$: \enskip $\rrec$ \\
			\quad \quad \quad $A$: \enskip $\ver_{\rrec} := (\guid(\tilde{\bidpreR}, \tilde{\uid}) \neq 1)$ \\

			$V$: \enskip \textbf{else if} $\rrec = (\alreadyreg, \tilde{\bidpreR}, \tilde{\bidR})$\textbf{:} \\
			\quad \quad \quad $V \rightarrow A$: \enskip $\rrec$ \\
			\quad \quad \quad $A$: \enskip $\ver_{\rrec} :=  (\glive(\tilde{\bidR}, \envR) = \gmatch(\tilde{\bidR}, \tilde{\bidpreR}) = 1)$ \\

			$V$: \enskip \textbf{else if} $\rrec = (\accept, \rrec_{11} = (\tilde{\uid}, \tilde{\bidR}, \tilde{\rho}_{\mathsf{enc},1}), \rrec_{12} = (\tilde{c}_{\vid}, \tilde{c}_{j}, \tilde{c}_{\ed}, \tilde{c}_{\bidR},$ \\
			\quad \enskip \ $\tilde{\rho}_{\mathsf{eq},1}), \rrec_{21} = (\tilde{\vid}, \tilde{j}, \tilde{\ed}, \tilde{\rho}_{\mathsf{enc},2}), \rrec_{22} = (\tilde{c}_{\vid}^*, \tilde{c}_j^*, \tilde{c}_{\ed}^*, \tilde{\rho}_{\mathsf{eq},2}))$\textbf{:} \\
			\quad \quad \quad $V$: \enskip $\beta \xleftarrow{\$} \{0,1,2\}$ \\
			\quad \quad \quad  $V \rightarrow A$: \enskip $(\rrec_{11},\rrec_{12})$ if $\beta=0$; $(\rrec_{21},\rrec_{22})$ if $\beta=1$; $(\rrec_{12},\rrec_{22})$ o/w \\
			\quad \quad \quad  $A$: \enskip $\ver_{\rrec} \leftarrow (\nizkver(\tilde{\rho}_{\mathsf{enc},1}, (\pk_S, \tilde{\bidR}, \tilde{c}_{\bidR})) \wedge$ \\
			\quad \quad \quad \quad \enskip \quad $\exists i: \ERR_i = (\tilde{\uid}, \tilde{c}_{\vid}, \tilde{c}_j, \tilde{c}_{\ed}, \tilde{c}_{\bidR}))$ if $\beta=0$; \\
			\quad \quad \quad  \quad \enskip $\nizkver(\tilde{\rho}_{\mathsf{enc},2}, (\pk_S, (\tilde{\vid}, \tilde{j}, \tilde{\ed}),(\tilde{c}_{\vid}^*,\tilde{c}_{j}^*, \tilde{c}_{\ed}^*)))$ if $\beta=1$;\\
			\quad \quad \quad  \quad \enskip $\nizkver(\tilde{\rho}_{\mathsf{eq},1} \oplus \tilde{\rho}_{\mathsf{eq},2}, (\pk_S, ((\tilde{c}_{\vid}, \tilde{c}_{\vid}^*), (\tilde{c}_{j},\tilde{c}_{j}^*), (\tilde{c}_{\ed}, \tilde{c}_{\ed}^*))))$ o/w \\

			\textbf{output} $\ver_{\rrec}$ \\
			\hline

			\\

			\hline
			$\univaudit(\mathsf{pk}_S, \alpha, \ERR=(\ve{\uid}, \cdot, \cdot, \cdot, \ve{c_{\bidR}}), \dots, (\si{S_k}{\mathsf{sk}_{S_k}})_{k \in [\kappa]},\si{A}{})$: \\
			\hline

			\textbf{assert} $\forall i,j\in |\ERR|, i\neq j: \uid_i \neq \uid_j$ \\
			$I_{\ERR\text{-}\mathsf{audit}} \leftarrow H_{\alpha}([|\ERR|])$ \\
			\textbf{for $i \in I_{\ERR\text{-}\mathsf{audit}}$:} \\
			\quad $\so{A}{\bidR_i, \cdot} \leftarrow \thsd(\pk_S, c_{\bidR_i}, (\si{S_k}{\sk_{S_k}})_{k \in [\kappa]}, \si{A}{})$ \\
			\quad $\ver_{\ERR_i} := (\guid(\bidR_i, \uid) = 1)$ \\
			$\ver_{\ERR} \leftarrow (\bigwedge_{i \in I_{\ERR\text{-}\mathsf{audit}}} \ver_{\ERR_i})$ \\
			$\dots$ \mycomment{contd. to Figures \ref{fig:protocol-prepareER} and \ref{fig:protocol-cast}} \\
			\hline
		\end{tabular} 
	}
	\caption{The $\setup$ protocol, $\register$ protocol, $\indrraudit$ protocol and $\ERR$ audit steps of the $\univaudit$ protocol}
	\label{fig:protocol-register}
\end{figure}

\subsubsection{Preventing denial via $\reject$ registration receipts}
\label{sec:registration-denial-reject}

Eligible voters may be denied from even registering by malicious ROs giving them a $\reject$ registration receipt, citing that they were trying to register using someone else's $\uid$. To prevent this, we require the RO to capture a live photograph $\bidpreR$ of the voter, use this photograph to verify if the voter indeed owns the presented $\uid$ and print $\uid$ and $\bidpreR$ on the $\reject$ receipt. The voter verifies the correctness of $\uid$ and the just-captured $\bidpreR$, whereas the receipt audit verifies whether the identifier corresponding to $\bidpreR$ indeed mismatches $\uid$ as per the primary identity system. It is important that $\bidpreR$ is captured in a special ``pre-registration'' environment $\envpreR$ distinct from the official registration environment $\envR$ used for actual registration to prevent denial with a claim of duplicate registration - see below.

\subsubsection{Preventing denial via $\alreadyreg$ registration receipts}
\label{sec:registration-denial-alreadyreg}

Eligible voters may also be denied by ROs issuing them a $\alreadyreg$ receipt and wrongly claiming that they have already registered before. Malicious ROs may try to do this by uploading bogus registration information to $\ERR$ before the voter's arrival in the registration booth. To prevent this, we require the RO to print the decryption $\bidR_i$ of $c_{\bidR_i}$ published against the $\uid$ from the claimed previous registration in the $\alreadyreg$ receipt, along with a live photograph of the voter in environment $\envpreR$. The voter verifies $\bidpreR$ and the receipt audit verifies whether $\bidR_i$ matches $\bidpreR$ and is indeed in environment $\envR$. Since voters do not give their $\envR$ photographs until they come for registration, malicious ROs cannot produce such $\bidR_i$s for them by Assumption \ref{assumption:UP}. Also note that malicious ROs may even try to deny registration to voters coming for registration and surreptitiously use their $\envR$ photograph instead for this purpose. This is prevented because voters do not accept a $\reject$ or $\alreadyreg$ receipt after giving a $\envR$ photograph.

\subsubsection{Preventing denial via $\accept$ registration receipts}
\label{sec:registration-denial-accept}

Even voters given an $\accept$ receipt could be denied if no registration request was actually uploaded to $\ERR$, or if the uploaded encryptions $c_{\bidR}$, $c_{\ed}$ and $c_{j}$ do not correctly encrypt voters' supplied information, resulting in dismissal of the voter's registration request by $E_j$. Further, if the $\vid$ on the voter's receipt does not match the $c_{\vid}$ uploaded against their $\uid$, they would be denied during vote casting as it would either make $P_l$ unable to find $\vid$ on $\ER$ or fetch some other voter's registered photograph. The problem is that a bare-handed voter cannot detect and contest against such incorrectly computed encryptions, and putting proofs of their correctness directly in the voter receipt leaks the voter's complete identity information to the auditor. 

To address this, we design a special cut-and-choose protocol. We require an $\accept$ registration receipt to consist of four components physically arranged as follows and separated from each other via perforations at the dashed lines: 
\begin{equation*}
\begin{tabular}{|l:l|} \hline $\rrec_{11} := {\color{blue}\uid}, {\color{blue}\bidR}, \rho_{\mathsf{enc},1}$ & $\rrec_{12} := c_{\vid}, c_j, c_{\ed}, c_{\bidR},\rho_{\mathsf{eq},1}$ \\ \hdashline  $\rrec_{21} := \vid, {\color{blue}j}, {\color{blue}\ed}, \rho_{\mathsf{enc},2}$ & $\rrec_{22} := c^*_{\vid}, c^*_j, c^*_{\ed}, \rho_{\mathsf{eq},2}$ \\ \hline \end{tabular} 
\end{equation*}
Here, $\rho_{\mathsf{enc},1}$ is a NIZK proof that $c_{\bidR}$ encrypts $\bidR$; $\rho_{\mathsf{enc},2}$ is a NIZK proof that $c^*_{\vid}$, $c^*_j$, and $c^*_{\ed}$ are encryptions of $\vid$, $j$ and $\ed$ (different from $c_{\vid}$, $c_j$ and $c_{\ed}$) respectively; and $\rho_{\mathsf{eq},1}$ and $\rho_{\mathsf{eq},2}$ are $(2,2)$-secret shares of a NIZK proof $\rho_{\mathsf{eq}}$ proving that $(c^*_{\vid}, c^*_j, c^*_{\ed})$ encrypt the same values as $(c_{\vid},c_j,c_{\ed})$. The voter is only required to verify the correctness of the plaintext $\uid$, their just-captured photograph $\bidR$, and the supplied information $j$ and $\ed$ printed on the receipt in a human-readable form (shown in {\color{blue}blue} above). For the receipt audit, the voter needs to randomly perform one of the following three steps: $1)$ tear off the receipt along the horizontal perforation line and give the top half to the auditor, $2)$ tear it off as above and give the bottom half to the auditor, and $3)$ tear it off along the vertical perforation line and give the right half to the auditor. If the uploaded encryptions $(c_{\vid}, c_j, c_{\ed}, c_{\bidR})$ are incorrect, one of the three verifications must fail. Thus, the approach statistically detects a large number of such frauds, where each fraud is detected with probability 1/3 on audit, but the auditor cannot link the $\uid$ and $\vid$ information for any voter. Note that voters can take a copy of the bottom left half of the registration receipt for vote casting.

\subsubsection{Preventing voter stuffing}
\label{sec:registration-stuffing}

Note that anyone can request to register for the election; the actual eligibility of voters is established by auditing EO approvals during the $\ER$ audit (see \S\ref{sec:casting-stuffing}). Thus, the only voter stuffing attack within the scope of the $\ERR$ universal audit is multiple registrations by a single voter. This is prevented by checking for the audited $\ERR$ entries that $c_{\bidR}$ encrypts a photograph of $\uid$'s owner and that each $\uid$ is distinct. By Assumption \ref{assumption:NDI}, this prevents a single voter's photographs from appearing in multiple $\ERR$ entries.

\subsubsection{Privacy issues}
\label{sec:registration-privacy}

Given our assumption of an honest $R_m$ for privacy and the need to protect privacy only for honest eligible voters (see \S\ref{sec:privacy}), we focus on the privacy leak via $\accept$ receipts for honest voters. As shown in \S\ref{sec:registration-denial-accept}, this case allows the auditor to learn only either the $\uid$-related or the $\vid$-related information but not both. 

However, malicious voters may try to extract honest voters' registration information by deliberately going to register with wrong data, knowing they will be denied by honest $R_m$, but hoping to learn extra information from the receipts. To prevent against this, we ensure that a voter can only ever fetch its own information from the $\register$ protocol and not of any other voter. Specifically, in the $\reject$ case, no information about $\uid$'s owner is revealed. An approach that, e.g., prints the decryption of $c_{\bidR}$ uploaded against $\uid$ --- to be matched against $\bidpreR$ during receipt audit --- would allow any voter to fetch any uploaded $\uid$'s registered photograph. In the $\alreadyreg$ case, the voter only obtains the decryption $\bidR_i$ of $c_{\bidR_i}$ after the RO checks the assertion that $\bidR_i$ matches $\bidpreR$. Note that because of this check, our protocol is complete only if $\guid$ is such that no two voters are given the same $\uid$ (our soundness guarantee does not require this assumption and only requires that no voter should hold two $\uid$s - Assumption \ref{assumption:NDI}.)

Finally, the $\ERR$ entries decrypted during $\univaudit$ are selected non-interactively via a hash function. Thus, only a fixed number of random entries are ever decrypted, even when the protocol is invoked multiple times by different auditors. 

\subsection{Preparation of the electoral roll (Figure \ref{fig:protocol-prepareER})}
\label{sec:prepareER}

\begin{figure}[t]
	\scalebox{0.8}{
		\begin{tabular}[t]{@{}l@{}}
			\hline
			$\prepareRoll(\pk_S, \ERR=(\ve{\uid}, \ve{c_{\vid}}, \ve{c_{j}}, \ve{c_{\ed}}, \ve{c_{\bidR}}), \si{S_k}{\sk_{S_k}})_{k \in [\kappa]}$, $((\si{E_j}{})_{j \in [\jmath]})$: \\
			\hline
			$(\ve{c_{\vid}'}, \ve{c_j'}, \ve{c_{\ed}'}, \ve{c_{\bidR}'}), \rho_{\mathsf{shuf}} \leftarrow \mathsf{Shuffle}(\mathsf{pk}_S, (\ve{c_{\vid}}, \ve{c_j}, \ve{c_{\ed}}, \ve{c_{\bidR}}), (\si{S_k}{})_{k \in [\kappa]})$ \\
			$(\ve{\vid'}, \ve{j'}), (\ve{\rho_{\vid'}}, \ve{\rho_{j'}}) \leftarrow \thsd(\mathsf{pk}_S, (\ve{c_{\vid}'}, \ve{c_{j}'}), (\si{S_k}{\sk_{S_k}})_{k \in [\kappa]})$ \\
		    \textbf{for $j \in [\jmath]$:} \\
			\quad \textbf{for $i$ s.t. $\vc{j}{i}' = j$:} \\
			\quad \quad $\so{E_j}{(\vc{\bidR}{i}',\vc{\ed}{i}'), \cdot} \leftarrow \thsd(\mathsf{pk}_S, (\vc{c_{\bidR}'}{i}, \vc{c_{\ed}'}{i}), (\si{S_k}{\sk_{S_k}})_{k \in [\kappa]}, \si{E_j}{})$ \\
			\quad \quad $E_j$: \enskip $\rresp_i \leftarrow \gelg(\bidR_i', \ed_i')$ \\
			\textbf{output} $\ER = (\ve{\vid}', \ve{j}', \ve{c_{\ed}'}, \ve{c_{\bidR}'}, \ve{\rresp}, \ve{\bot})$, $\rho_{\ER} = (\rho_{\mathsf{shuf}}, \ve{\rho_{\vid'}}, \ve{\rho_{j'}}, \ve{c_{\vid}'}, \ve{c_{j}'})$ \\
			\hline
			\\
			\hline
			$\univaudit(\mathsf{pk}_S, \cdot, \ERR = (\ve{\uid}, \ve{c_{\vid}}, \ve{c_j}, \ve{c_{\ed}}, \ve{c_{\bidR}}), \ER = (\ve{\vid'}, \ve{j'}, \ve{c_{\ed}'},\ve{c_{\bidR}'},\cdot, \cdot),$ \\
			\quad $\rho_{\ER} = (\rho_{\mathsf{shuf}}, \ve{\rho_{\vid'}}, \ve{\rho_{j'}}, \ve{c_{\vid}'}, \ve{c_{j}'}), \cdot, \cdot, (\si{S_k}{\mathsf{sk}_{S_k}})_{k \in [\kappa]}, \si{A}{})$: \\
			\hline
			$\dots$ \mycomment{contd. from Figure \ref{fig:protocol-register}} \\
			$A$: \enskip \   $\ver_{\prepareRoll} := (\forall i,j \in [|\ER|],i \neq j: \vid'_i \neq \vid'_j) \wedge$ \\
			\quad \quad \quad $\nizkver(\rho_{\mathsf{shuf}}, (\mathsf{pk}_S, (\ve{c_{\vid}}, \ve{c_j}, \ve{c_{\ed}}, \ve{c_{\bidR}}), (\ve{c'_{\vid}}, \ve{c'_j}, \ve{c'_{\ed}}, \ve{c'_{\bidR}}))) \wedge$ \\
			\quad \quad \quad $\nizkver(\ve{\rho_{\vid'}}, (\pk_S, \ve{\vid'}, \ve{c'_{\vid}})) \wedge \nizkver(\ve{\rho_{j'}}, (\pk_S, \ve{j'}, \ve{c'_{j}}))$ \\
			$\dots$ \mycomment{contd. to Figure \ref{fig:protocol-cast}} \\
			\hline
		\end{tabular} 
	}
	\caption{The $\prepareRoll$ protocol and $\prepareRoll$ audit steps of the $\univaudit$ protocol}
	\label{fig:protocol-prepareER}
\end{figure}

\subsubsection{Preventing voter stuffing} 
\label{sec:prepareER-stuffing}

The proof of shuffle of list $\ve{c_{\bidR}}$ to $\ve{c'_{\bidR}}$ ensures that if no single voter's registration photographs appear in multiple $\ve{c_{\bidR}}$ entries (see \S\ref{sec:registration-stuffing}) then they do not appear in multiple $\ve{c'_{\bidR}}$ entries. This prevents a single voter from owning multiple $\ER$ entries.

\subsubsection{Preventing denial} 
\label{sec:prepareER-denial}

Additionally, proofs of consistent shuffle of lists $\ve{c_{\vid}}$, $\ve{c_{j}}$, $\ve{c_{\ed}}$ and $\ve{c_{\bidR}}$ under the same secret permutation and correct decryption of $\ve{c'_{\vid}}$ and $\ve{c'_{j}}$ ensures that accepted voters' supplied registration information correctly appears in encrypted form against the $\vid$ in their registration receipt. This prevents evaluation of their eligibility on wrong information or denial during casting with a ``wrong $\vid$'' claim.

\subsubsection{Privacy issues} 
\label{sec:prepareER-privacy}

The privacy guarantee of this step is mainly due to the unlinkability of shuffle input and output lists if at least one $S_k$ is honest. Further, the honest $S_{k^*}$ allows each $E_j$ to fetch only entries marked with block number $j$ and such that the decryption shares are directly sent to $E_j$, disallowing any $(S_k)_{k \in [\kappa]}$ from learning the decryptions themselves.

\subsection{Vote casting (Figure \ref{fig:protocol-cast})}
\label{sec:casting}

\begin{figure}[t]
	\scalebox{0.8}{
		\begin{tabular}[t]{@{}l@{}}
			\hline
			$\cast(\pk_S, \ER = (\ve{\vid'}, \ve{j'}, \ve{c_{\ed}'},\ve{c_{\bidR}'},\rresp_i, l_i), \CI = (\ve{c_{\bidC}}, \ve{c_{\rho_{\ev}}}), \ve{\ev},$ \\
			\quad $\si{V}{\vid, \ev, \rho_{\ev}}, \si{P_{l}}{ }, (\si{S_k}{\sk_{S_k}})_{k \in [\kappa]})$: \\
			\hline
			$V \rightarrow P_l$: \enskip $\vid,\ev,\rho_{\ev}$ \quad \quad \quad \quad \quad \quad \quad \mycomment{$\ev,\rho_{\ev}$ obtained from $\pietoev.\cast$}\\
			
			$P_l$: \enskip \textbf{if} $\forall i: \vid'_i \neq \vid$\textbf{:} \quad \quad \quad \quad \quad \quad \quad \quad \quad \mycomment{Rejected - unrecognised $\vid$} \\
			\quad \quad \quad $P_l \rightarrow V$: \enskip $\crec := (\reject, \vid)$ \\
			\quad \quad \quad $V$: \enskip $\cvalid := 1$ iff $\crec := (\reject, \vid)$ \\

			$P_l$: \enskip \textbf{else if} $\exists i: \vid'_i = \vid$\textbf{:} \\
			\quad \quad \quad $V,P_l$: \enskip $\bidpreC \leftarrow \gcapture(V, \envpreC)$ \\
			\quad \quad \quad $\so{P_l}{(\bidR'_i, \ed'_i, \bidC_i, \rho_{\ev_i}), (\rho_{\bidR'_i}, \rho_{\ed'_i}, \cdot, \cdot)} \leftarrow \thsd(\pk_S, (c'_{\bidR_i}, c'_{\ed_i},$ \\
			\quad \quad \quad \quad \quad $c_{\bidC_i}, c_{\rho_{\ev_i}}), (\si{S_k}{\sk_{S_k}})_{k \in [\kappa]}, \si{P_l}{})$ \\
			
			\quad \quad \quad \textbf{if} $\gmatch(\bidR'_i,\bidpreC) \neq 1$\textbf{:} \quad \quad \quad \quad \quad \quad \quad \mycomment{Rejected - incorrect $\vid$} \\
			\quad \quad \quad \quad $P_l \rightarrow V$: \enskip $\crec := (\reject, \vid, \bidpreC)$ \\
			\quad \quad \quad \quad $V$: \enskip $\cvalid := 1$ iff $\crec := (\reject, \vid, \bidpreC)$ \\

			\quad \quad \quad \textbf{else if} $\rresp_i \neq 1$\textbf{:} \quad \quad \quad \quad \quad \quad \quad \quad \mycomment{Rejected - deemed ineligible} \\
			\quad \quad \quad \quad $P_l \rightarrow V$: \enskip $\crec := (\reject, \vid, \bidpreC, \ed'_i, \rho_{\ed'_i})$ \\
			\quad \quad \quad \quad $V$: \enskip $\cvalid := 1$ iff $\crec := (\reject, \vid, \bidpreC, \cdot, \cdot)$ \\
			
			\quad \quad \quad \textbf{else if} $\pietoev.\valid(\rho_{\ev_i}, \ev_i)$\textbf{:} \quad \quad \enskip \mycomment{Rejected - duplicate casting} \\
			\quad \quad \quad \quad $P_l \rightarrow V$: \enskip $\crec := (\alreadycast, \bidpreC, \bidC_i)$ \\
			\quad \quad \quad \quad $V$: \enskip $\cvalid := 1$ iff $\crec := (\alreadycast, \bidpreC, \cdot)$ \\

			\quad \quad \quad \textbf{else:} \quad \quad \quad \quad \quad \quad \quad \quad \quad \quad \quad \quad \quad \quad \quad \quad \mycomment{Accepted for casting} \\
			\quad \quad \quad \quad $V,P_l$: \enskip $\bidC \leftarrow \gcapture(V, \envC)$ \\
			\quad \quad \quad \quad $c_{\bidC} \leftarrow \thse(\pk_S, \bidC)$; $c_{\rho_{\ev}} \leftarrow \thse(\pk_S, \rho_{\ev})$ \\
			\quad \quad \quad \quad $\CI_i := (c_{\bidC},c_{\rho_{\ev}})$; $\ev_i := \ev$ \\
			\quad \quad \quad \quad $P_l \rightarrow V$: \enskip $\crec := (\accept, \vid, \ev)$ \\
			\quad \quad \quad \quad $V$: \enskip $\cvalid := 1$ iff $\crec := (\accept, \vid, \ev)$ \\
			\textbf{output} $\CI, \ve{\ev}, \so{V}{\crec,\cvalid}$, $\so{P_l}{}$ \\
			\mycomment{Post polling, for each $\ER_i$ where $l_i=l$ but no $\ev_i$ exists, $P_l$ uploads} \\
			\mycomment{$\CI_i = (\thse(\pk_S, \bot),\thse(\pk_S, \bot))$ and a fresh dummy $\ev_i$.} \\
			\hline
			\\

			\hline
			$\indcraudit(\mathsf{pk}_S, \ER = (\ve{\vid}', \dots), \CI, \ve{\ev},  \si{V}{\crec, \rrec},(\si{S_k}{\mathsf{sk}_{S_k}})_{k \in [\kappa]}, \si{A}{})$: \\
			\hline
			$V \rightarrow A$: \enskip $\crec$ \\  
			$A$: \  \textbf{if} $\crec = (\reject, \tilde{\vid})$\textbf{:} $\ver_{\crec}:= 1$ iff $\forall i: \vid'_i \neq \tilde{\vid}$ \\

			\quad \enskip \textbf{else if} $\crec = (\reject, \tilde{\vid}, \tilde{\bidpreC})$\textbf{:} \\
			\quad \enskip \quad $V \rightarrow A$: \enskip $(\rrec_{11},\rrec_{21}) = ((\cdot, \bidR, \cdot), (\vid, \cdot, \cdot, \cdot))$ \\ 
			\quad \enskip \quad $\ver_{\crec}:= \vid \neq \tilde{\vid} \vee \gmatch(\tilde{\bidpreC}, \bidR) = 0$ \\

			\quad \enskip \textbf{else if} $\crec = (\reject, \tilde{\vid}, \tilde{\bidpreC}, \tilde{\ed}', \tilde{\rho}_{\ed'})$\textbf{:} \\
			\quad \enskip \quad $\ver_{\crec}:= \exists i: \vid'_i = \tilde{\vid} \wedge \mathsf{NIZKVer}(\tilde{\rho}_{\ed'}, (\pk_S, \tilde{\ed}', c'_{\ed_i})) \wedge$ \\
			\quad \enskip \quad \quad $\gelg(\tilde{\bidpreC}, \tilde{\ed}') = 0$ \\
			
			\quad \enskip \textbf{else if} $\crec = (\alreadycast, \tilde{\bidpreC}, \tilde{\bidC})$\textbf{:} \\
			\quad \enskip \quad $\ver_{\crec}:= \gmatch(\tilde{\bidC}, \tilde{\bidpreC}) = \glive(\tilde{\bidC}, \envC) = 1$ \\
			
			\quad \enskip \textbf{else if} $\crec = (\accept, \tilde{\vid}, \tilde{\ev})$\textbf{:} \\
			\quad \enskip \quad $\ver_{\crec}:= (\exists i: \vid'_i = \tilde{\vid} \wedge \ev_i = \tilde{\ev})$ \\
			\textbf{output} $\ver_{\crec}$ \\
			\hline
			\\

			\hline
			$\univaudit(\mathsf{pk}_S$, $\alpha$, $\RER$, $\ER = (\ve{\vid'}, \ve{j'},\ve{c_{\bidR}'}, \ve{c_{\ed}'},\ve{\rresp}, \cdot)$, $\cdot$, $\CI = (\ve{c_{\bidC}}, \ve{c_{\rho_{\ev}}})$, \\
			\quad $\ve{\ev}$, $(\si{S_k}{\mathsf{sk}_{S_k}})_{k \in [\kappa]}$, $\si{A}{})$: \\
			\hline
			$\ver_{\ERR}, \ver_{\prepareRoll} \leftarrow \dots$ \mycomment{contd. from Figures \ref{fig:protocol-register} and \ref{fig:protocol-prepareER}} \\
			$I_{\ER\text{-}\mathsf{audit}} \leftarrow H_{\alpha}(\{ i \in [|\ER|] \mid \rresp_i = 1 \})$ \\
			\textbf{for $i \in I_{\ER\text{-}\mathsf{audit}}$:} \\
			\quad \textbf{if} $\pietoev.\valid(\rho_{\ev_i},\ev_i)=1$\textbf{:} \\
			\quad \quad $\ver_{\ER_i} := (\gmatch(\bidC_i, \bidR'_i) = \glive(\bidC_i, \envC) = \gelg(\bidR'_i, \ed'_i) = 1)$, \\
			\quad where each $d \in \{\bidR'_i, \bidC_i, \ed'_i, \rho_{\ev_i}\}$ is obtained as follows\textbf{:} \\
			\quad \quad $\so{A}{d, \cdot} \leftarrow \thsd(\pk_S, c_d, \si{S_k}{\sk_{S_k}})_{k \in [\kappa]}, \si{A}{}) $ \\
			$\ver_{\ER} \leftarrow \bigwedge_{i \in I_{\ER\text{-}\mathsf{audit}}} \ver_{\ER_i}$ \\
			\textbf{output} $\ver_{\mathsf{UA}} := \ver_{\ERR} \wedge \ver_{\prepareRoll} \wedge \ver_{\ER}$ \\
			\hline
		\end{tabular}
	}
	\caption{The $\cast$ protocol, $\indcraudit$ protocol and $\ER$ audit steps of the $\univaudit$ protocol}
	\label{fig:protocol-cast}
\end{figure}

\subsubsection{Preventing denial via $\reject$ cast receipts}
\label{sec:casting-denial-reject}

Voters could be denied by malicious POs via a $\reject$ cast receipt with multiple claims: $a)$ they used a $\vid$ not in $\ER$, $b)$ the $\vid$ was assigned to someone else, or $c)$ the voter was deemed ineligible by the EO. In case $a$, the receipt simply contains the $\vid$ used and can be directly audited against $\vid$s published on $\ER$. In case $b$, we require the PO to match the $\bidR'$ decrypted from $c_{\bidR'}$ published against the presented $\vid$ entry, match it against a live photograph $\bidpreC$ of the voter in a ``pre-casting'' environment $\envpreC$, and issue a receipt containing $\bidpreC$ and $\vid$. If the voter is unfairly denied, they can prove ownership of $\vid$ by showing the left half of their registration receipt (see \S\ref{sec:registration-denial-accept}) to the auditor, who verifies that the $\bidR$ in the registration receipt matches $\bidpreC$ in the cast receipt. In case $c$ (the voter is the owner of $\vid$), the PO issues a receipt containing the decrypted eligibility data registered against the presented $\vid$ entry, along with a proof of correct decryption. The receipt can then be audited directly using the $\gelg$ oracle. 

\subsubsection{Preventing denial via $\alreadycast$ cast receipts}
\label{sec:casting-denial-alreadycast}

Voters denied with the claim of duplicate casting can be issued an $\alreadycast$ receipt, similar to the $\alreadyreg$ receipt (see \S\ref{sec:registration-denial-alreadyreg}), containing a live photograph $\bidpreC$ of the voter in environment $\envpreC$ and a photo $\bidC$ in the official casting environment $\envC$ fetched from the claimed previous successful vote cast attempt. The protection mechanisms are similar to \S\ref{sec:registration-denial-alreadyreg}.

\subsubsection{Preventing denial via $\accept$ cast receipts}
\label{sec:casting-denial-accept}

Voters accepted during vote casting can be denied by giving them wrong $\vid$ or $\ev$ in the receipt or by not uploading their desired $\ev$. To prevent this, we require that the $\accept$ receipt should contain $(\vid,\ev)$ such that $\vid$ is affixed on the underlying E2E-V protocol's paper receipt containing $\ev$ by the PO. This is done like so because although bare-handed voters can detect manipulations in $\vid$ if it comes from a reasonably small space, they cannot detect if any modifications are made to $\ev$ by the PO.
Wrongly uploaded $\ev$s are prevented by a few random voters verifying that the $\ev$ on their receipt is correctly uploaded against $\vid$, exactly how traditional E2E-V protocols \citep{pretavoter,scratchandvote,starvote} achieve individual verifiability.

\subsubsection{Preventing voter/vote stuffing} 
\label{sec:casting-stuffing}

Preventing vote stuffing in this step requires counting votes from only eligible voters who came to cast their vote and counting them only once per voter. Votes by ineligible voters are prevented by checking during the universal audit of random $\ER$ entries that if the corresponding $\ev$ is a valid one (by decrypting and verifying the NIZK proof of validity $\rho_{\ev}$), then the voter is eligible as per their registered photograph and supplied eligibility data. Votes against absentee eligible voters are prevented by checking that the casting-stage photograph decrypted from $c_{\bidC}$ matches their registered photograph and was captured in environment $\envC$. Since absentee voters never give their photographs in environment $\envC$ and voters coming for casting do not accept a $\reject$ or $\alreadycast$ receipt after giving a $\envC$ photograph, this check ensures that the correctness of uploaded $\ev$'s is auditable by the voters. Multiple casting attempts by the same voter are prevented since if no voter's registered photograph appears in multiple $\ER$ entries (\S\ref{sec:prepareER-stuffing}), each voter's casting photograph matches at most one $\ER$ entry's registered photograph.

\subsubsection{Privacy issues}
\label{sec:casting-privacy}

As with \S\ref{sec:registration-privacy}, we focus on the case of an $\accept$ receipt for honest voters. The secrecy of this receipt directly follows from the receipt-freeness of the underlying E2E-V scheme. Further, as in registration, no voter obtains decryption of any uploaded $\ER$ entry if their live photograph does not match the registered photograph published against the claimed $\vid$. The universal audit similarly only decrypts a few random entries selected by a hash function.

\subsection{Additional issues}
\label{sec:ai}

\subsubsection{Authentication of the receipts} 
\label{sec:ai-receipt-authentication}

All receipts obviously need to be authenticated with signatures of appropriate authorities to prevent frivolous denial claims by malicious voters. However, one must be careful that this does not introduce new opportunities of denial to honest voters. Specifically, if a receipt is authenticated via a \emph{digital} signature, $R_m$ and $P_l$ can deny voters by deliberately issuing invalid signatures since bare-handed voters cannot verify them but they render the voters' denial claims frivolous. Thus, all receipts must be authenticated using traditional seals-and-stamps verifiable by human voters. This, of course, comes with the unavoidable trade-off of reduced security against false denial claims.

\subsubsection{Space of $\uid$ and $\vid$} 
\label{sec:ai-vid-uid-verification}

The space of $\uid$ and $\vid$ must be small enough so that bare-handed voters can detect modified $\uid/\vid$ printed on their receipts but big enough to assign each voter a unique identifier.

\subsubsection{Forced abstention from registration}
\label{sec:ai-forced-abstention-registration}

Since $\uid$s of voters participating in registration become public, the protocol presented as above does not protect against adversaries forcing voters to abstain from registration. Further, participation in registration for an election may itself leak information about the voter's residential constituency and other attributes. To prevent this, $R_m$ can upload dummy encrypted registration information for all $\uid$s enrolled in the primary identity system whose owner did not show up for registration. This is allowed because of our guarantees against voter stuffing.



\subsubsection{Voters losing or forgetting $\vid$s}

Losing one's $\vid$ to someone else does not enable them to vote on the first voter's behalf because of photo matchings. However, if a voter forgets their $\vid$, they cannot vote. Since $E_j$'s are responsible for enfranchising voters and hold their registered photographs, such voters should fetch their $\vid$ from their $E_j$. 

\subsubsection{Dynamic updates and verifications}
\label{sec:ai-dynamic}

Finally, note that our protocol works in a simplified static model where voters supply registration information only once per election, and the $\ERR$ list is published and verified only once. In practice, though, eligibililty data is updated and verified dynamically throughout the election cycle. To allow voters to update their eligibility data without allowing ROs to maliciously change voters' data, all photographs must show a unique timestamp that the uploaded registration information must also carry. Any photograph should only authorise information marked with the same timestamp. To prevent later modification of entries once-verified, an \emph{append-only bulletin board} \citep{BB0k} can be used \citep{bulletin-board}.


\section{Security analysis}
\label{sec:analysis}

We begin our analysis by first defining a few primitive concepts and proving some basic lemmas. We then state our main soundness claim in Theorem~\ref{thm:soundness}. Our main privacy claim is in Theorem~\ref{thm:privacy}.

\begin{definition}[Mapping between $\ERR$ and $\ER$ entries]
	\label{def:mapping-entries}
	An $\ERR_i$ entry \emph{maps to} an $\ER_{i'}$ entry (equivalently, $\ER_{i'}$ maps from $\ERR_i$) if $a)$ $\vid'_{i'}$ and $j'_{i'}$ are the correct decryptions of $c_{\vid_i}$ and $c_{j_i}$ respectively, and $b)$ the correct decryptions of $c_{\ed_{i'}}'$ and $c_{\bidR_{i'}}'$ equal the correct decryptions of $c_{\ed_i}$ and $c_{\bidR_i}$ respectively.
\end{definition}

\begin{lemma}
	\label{lem:correct-shuffle}
	If $\ver_{\prepareRoll}=1$ then each $\ERR$ entry maps to one and only one $\ER$ entry and each $\ER$ entry maps from one and only one $\ERR$ entry.
\end{lemma}
\begin{proof}
	Since $\ver_{\prepareRoll}=1$, NIZK $\rho_{\mathsf{shuf}}$ passes (see Figure \ref{fig:protocol-prepareER}). Thus, a permutation $\pi$ can be extracted such that lists $\ve{c_{\vid}'}$, $\ve{c_{j}'}$, $\ve{c_{\ed}'}$ and $\ve{c_{\bidR}'}$ are re-encryptions and permutations of lists $\ve{c_{\vid}}$, $\ve{c_{j}}$, $\ve{c_{\ed}}$ and $\ve{c_{\bidR}}$ respectively under the same permutation $\pi$. Further, NIZKs $\ve{\rho_{\vid'}}$ and $\ve{\rho_{j'}}$ also pass, so $\ve{\vid}'$ and $\ve{j}'$ correctly decrypt $\ve{c_{\vid}'}$ and $\ve{c_{j}'}$ respectively. Thus, each $\ERR_i$ entry maps to $\ER_{\pi(i)}$ and each $\ER_i$ entry maps from $\ERR_{\pi^{-1}(i)}$. If some $\ERR_i$ entry maps to $\ER_{i_1}$ and $\ER_{i_2}$ for two distinct $i_1$ and $i_2$, then both $\vc{\vid}{i_1}'$ and $\vc{\vid}{i_2}'$ must be correct decryptions of $c_{\vid_i}$. By the correctness of the decryption protocol, this implies that $\vc{\vid}{i_1}' = \vc{\vid}{i_2}'$. This is not possible because $\ver_{\prepareRoll}=1$ only if all $\vid$s are distinct. Thus, each $\ERR_i$ entry must map to only $\ER_{\pi(i)}$. This also implies that $\ERR_{i_1}$ and $\ERR_{i_2}$ for two distinct $i_1$ and $i_2$ cannot map to a single $\ER_i$ entry, because no permutation $\pi$ exists such that $\pi(i_1) = \pi(i_2) = i$.
\end{proof}

\begin{fact}
	\label{fact:photo-captures-one}
	Each photograph $p$ captures at most one voter $V \in \Voters$ (directly follows from Assumptions \ref{assumption:CP} and \ref{assumption:NT} by setting $p':=p$).
\end{fact}

\begin{definition}[Ownership]
	\label{def:ownership}
	Let $V_i \in \Voters$ be a voter s.t. $i \in I_{\h}$ in $\mathsf{Exp}_{\mathsf{soundness}}$ (Figure \ref{fig:exp-soundness}). \textbf{(1)} $V_i$ \emph{owns} the registration receipt $\rrec_i^{\r}$ and the cast receipt $\crec_i^{\c}$. \textbf{(2)} $V_i$ owns an $\ERR_{i^*}$ entry if the decryption $\bidR_{i^*}$ of $c_{\bidR_{i^*}}$ captures $V_i$ and $\guid(\bidR_{i^*}, \uid_{i^*}) = 1$. $V_i$ owns an $\ER_{i^{+}}$ entry if $\ER_{i^{+}}$ maps from an $\ERR_{i^*}$ entry owned by $V_i$.
\end{definition}

\begin{lemma}
	\label{lem:ownership}
	If $\ver_{\prepareRoll}=1$, then \textbf{(1)} each registration receipt, cast receipt, $\ERR$ entry and $\ER$ entry is owned by at most one voter in $V \in (V_{i})_{i \in I_{\h}}$; and \textbf{(2)} each voter in $\Voters$ owns at most one $\ERR$ entry and at most one $\ER$ entry.
\end{lemma}
\begin{proof}
	\textbf{(1)} Consider two voters $V_{i_1}$ and $V_{i_2}$ for $i_1,i_2 \in I_{\h}$. By Definition \ref{def:ownership}, the registration receipt and cast receipt obtained by $V_{i_1}$ and $V_{i_2}$ are owned by themselves and not by both. $V_{i_1}$ and $V_{i_2}$ cannot own a common $\ERR_{i^*}$ entry because by Fact \ref{fact:photo-captures-one}, the decryption of $c_{\bidR_{i^*}}$ captures at most one of $V_{i_1}$ and $V_{i_2}$. Now suppose $V_{i_1}$ and $V_{i_2}$ both own a common $\ER_{i^{+}}$ entry. This means that $\ER_{i^{+}}$ maps from an $\ERR_{i_1^*}$ entry owned by $V_{i_1}$ and from an $\ERR_{i_2^*}$ entry owned by $V_{i_2}$. Since $V_{i_1}$ and $V_{i_2}$ cannot own a common $\ERR_{i^{*}}$ entry, $i_1^* \neq i_2^*$. This means that $\ER_{i^{+}}$ maps from two distinct $\ERR$ entries, which violates Lemma \ref{lem:correct-shuffle}. 

	\textbf{(2)} Suppose some voter $V \in \Voters$ owns $\ERR_{i_1^*}$ and $\ERR_{i_2^*}$ for some $i_1^{*} \neq i_2^{*}$. Then both $\bidR_{i_1^*}$ and $\bidR_{i_2^*}$ capture $V$. Thus, by Assumption \ref{assumption:CP}, $\gmatch(\bidR_{i_1^*},\bidR_{i_2^*})=1$. $V$'s ownership of $\ERR_{i_1^*}$ and $\ERR_{i_2^*}$ also implies that $\guid(\bidR_{i_1^*}, \uid_{i_1^*}) = 1$ and $\guid(\bidR_{i_2^*}, \uid_{i_2^*}) = 1$. Thus, by Assumption \ref{assumption:NDI}, $\uid_{i_1^*}=\uid_{i_2^*}$. However, this is not possible because $\ver_{\prepareRoll}=1$ only if $\uid_{i_1^*} \neq \uid_{i_2^*}$. Since each $\ERR$ entry maps to only one $\ER$ entry, $V$ owns at most one $\ER$ entry too.
\end{proof}

\begin{definition}[Items and bad items]
	\label{def:bad-items}
	An item is a registration receipt, a cast receipt, an $\ERR$ entry or an $\ER$ entry. A registration receipt $\rrec$ is \emph{bad} if auditing it during $\indrraudit$ may result in $\ver_{\rrec}=0$. A cast receipt $\crec$ is \emph{bad} if auditing it during $\indcraudit$ results in $\ver_{\crec}=0$. An $\ERR_i$ (resp. $\ER_{i}$) entry is \emph{bad} if auditing it during $\univaudit$ results in $\ver_{\ERR_i}=0$ (resp. $\ver_{\ER_i}=0$).
\end{definition}

\begin{lemma}
	\label{lem:denial-implies-bad-items}
	If $\ver_{\prepareRoll}=1$ and some $V_i$ in $(V_i)_{i \in I_{\ch}}$ does not own an $\ER_{i^{+}}$ entry s.t. $\vc{\ev}{i^{+}} = \ev_i^{\c}$, then $V_i$ owns a bad receipt.
\end{lemma}
\begin{proof}
	As per $\mathsf{Exp}_{\mathsf{soundness}}$, for each $i \in I_{\ch}$, there exists a unique $\oreg(i,\dots)$ call and a unique $\ocast(i,\dots)$ call. Also, $\vid_i^{\c}=\vid_i^{\r}$, $\gelg(p_i^{\c},\ed_i^{\r})=1$, $\cvalid_i^{\c}=1$ and $\rvalid_i^{\r}=1$. We first prove that if $V_i$ does not own an $\ER_{i^{+}}$ entry where $\vid'_{i^{+}}=\vid_i^{r}$, $c'_{\ed_{i^{+}}}$ encrypts $\ed_i^{\r}$ and $c'_{\bidR_{i^{+}}}$ encrypts a photograph capturing $V_i$ in environment $\envR$, then the registration receipt $\rrec_i^{\r}$ owned by $V_i$ is bad:
	\begin{itemize}[leftmargin=*]
		\item \emph{Case 1: $\rrec_i^{\r} = (\reject, \tilde{\uid}, \tilde{\bidpreR})$.} 
			Since $\rvalid_i^{\r}=1$, $\tilde{\bidpreR}$ captures $V_i$ and $\tilde{\uid} = \uid_i^{\r}$. Since $\uid_i^{\r}$ satisfies $\guid(p_i^{\r}, \uid_i^{\r})=1$, we have $\guid(p_i^{\r}, \tilde{\uid})=1$. Since both $p_i^{\r}$ and $\tilde{\bidpreR}$ capture $V_i$, by Assumption \ref{assumption:CI}, $\guid(\tilde{\bidpreR}, \tilde{\uid})=1$. Thus, $\rrec_i^{\r}$ is bad.
		
		\item \emph{Case 2: $\rrec_i^{\r} = (\alreadyreg, \tilde{\bidpreR}, \tilde{\bidR})$.} 
			Since $\rvalid_i^{\r}=1$, $\tilde{\bidpreR}$ captures $V_i$. Further, it must be that $\gmatch(\tilde{\bidR}, \tilde{\bidpreR})=1$ and $\glive(\bidR, \envR)=1$, otherwise $\rrec_i^{\r}$ is bad. However, this is not possible by Assumption \ref{assumption:UP} since $V_i$ never called $\gcapture$ in $\envR$ in this case.
		
		\item \emph{Case 3: $\rrec_i^{\r} = (\accept$, $\rrec_{11} = (\tilde{\uid}$, $\tilde{\bidR}$, $\tilde{\rho}_{\mathsf{enc},1})$, $\rrec_{12} = (\tilde{c}_{\vid}$, $\tilde{c}_{j}$, $\tilde{c}_{\ed}$, $\tilde{c}_{\bidR}$, $\tilde{\rho}_{\mathsf{eq},1})$, $\rrec_{21} = (\tilde{\vid}$, $\tilde{j}$, $\tilde{\ed}$, $\tilde{\rho}_{\mathsf{enc},2})$, $\rrec_{22} = (\tilde{c}_{\vid}^*$, $\tilde{c}_j^*$, $\tilde{c}_{\ed}^*$, $\tilde{\rho}_{\mathsf{eq},2})$.} 
			We show that if $\rrec_i^{\r}$ is not bad then the required $\ER$ entry exists. It must be that an entry $\ERR_{i^*} = (\tilde{\uid}, \tilde{c}_{\vid}, \tilde{c}_{j}, \tilde{c}_{\ed}, \tilde{c}_{\bidR})$ exists and $\tilde{c}_{\vid}$, $\tilde{c}_{j}$, $\tilde{c}_{\ed}$ and $\tilde{c}_{\bidR}$ respectively encrypt $\tilde{\vid}$, $\tilde{j}$, $\tilde{\ed}$ and $\tilde{\bidR}$, otherwise $\ver_{\rrec_i^{\r}}=0$ if the bad receipt component is given to $A$ by $V_i$ and thus $\rrec_i^{\r}$ is bad by Definition \ref{def:bad-items}. Further, since $\rvalid_i^{\r}=1$, $\tilde{\bidR}$ captures $V_i$ in $\envR$, $\tilde{\ed}=\ed_i^{\r}$ and $\tilde{\uid}=\uid_i^{\r}$ and thus $\guid(p_i^{\r}, \tilde{\uid})=1$. Since both $p_i^{\r}$ and $\tilde{\bidR}$ capture $V_i$, by Assumption \ref{assumption:CI}, $\guid(\tilde{\bidR}, \tilde{\uid})=1$. Thus, by Definition \ref{def:ownership}, $V_i$ owns $\ERR_{i^*}$ and also the $\ER_{i^+}$ entry mapped from $\ERR_{i^*}$. Since $\ER_{i^+}$ maps from $\ERR_{i^*}$, by Definition \ref{def:mapping-entries}, $\vid'_{i^+} = \tilde{\vid} = \vid_i^{\r}$ and $c'_{\ed_{i^+}}$ and $c'_{\bidR_{i^+}}$ encrypt $\tilde{\ed}=\ed_i^{\r}$ and $\tilde{\bidR}$. $\ER_{i^+}$ is thus the desired entry.
	\end{itemize}

	Now, we show that if $V_i$ owns such an $\ER_{i^{+}}$ entry but $\vc{\ev}{i^{+}} \neq \ev_i^{\c}$, then the cast receipt $\crec_i^{\c}$ owned by $V_i$ is bad:
	\begin{itemize}[leftmargin=*]
		\item \emph{Case 1: $\crec_i^{\r} = (\reject, \tilde{\vid})$.} 
			Since $\cvalid_i^{\c}=1$, $\tilde{\vid}=\vid_i^{\c}=\vid_i^{\r}=\vid'_{i^+}$. Thus, $\crec_i^{\c}$ is bad.
		
		\item \emph{Case 2: $\crec_i^{\r} = (\reject, \tilde{\vid}, \tilde{\bidpreC})$.} 
			Since $\cvalid_i^{\c}=1$, $\tilde{\vid}=\vid_i^{\c}=\vid_i^{\r}=\vid$, where $\vid$ denotes the identity printed in $\rrec_i^{\r}$, and $\tilde{\bidpreC}$ captures $V_i$. Further, since $\rvalid_i^{\r}=1$, $\bidR$ printed in $\rrec_i^{\r}$ captures $V_i$. Thus, by Assumption \ref{assumption:CP}, $\gmatch(\bidR, \tilde{\bidpreC})=1$. Thus, Thus, $\crec_i^{\c}$ is bad.

		\item \emph{Case 3: $\crec_i^{\r} = (\reject, \tilde{\vid}, \tilde{\bidpreC}, \tilde{\ed}', \tilde{\rho}_{\ed'})$.} 
			Since $\cvalid_i^{\c}=1$, $\tilde{\vid} = \vid'_{i^+}$ and $\tilde{\bidpreC}$ captures $V_i$. Also, there must exist an $i^{++}$ s.t. $\vid'_{i^{++}} = \tilde{\vid}$, otherwise $\crec_i^{\r}$ is bad. Since all $\vid$s in $\ER$ are distinct, $i^{++} = i^+$. Thus, by the proof $\tilde{\rho}_{\ed'}$, $\tilde{\ed}'$ is the correct decryption of $c'_{\ed_{i^+}}$, i.e., $\tilde{\ed}' = \ed_i^{\r}$, otherwise $\crec_i^{\r}$ is bad. Since $\gelg(p_i^{\r}, \ed_i^{\r})=1$, by Assumption \ref{assumption:CE}, $\gelg(\tilde{\bidpreC}, \tilde{\ed}')=1$. Thus, $\crec_i^{\r}$ is bad.
		
		\item \emph{Case 4: $\crec_i^{\r} = (\alreadycast, \tilde{\bidpreC}, \tilde{\bidC})$.} 
			Since $\cvalid_i^{\c}=1$, $\tilde{\bidpreC}$ captures $V_i$. Thus, by Assumption \ref{assumption:UP}, it cannot be that $\gmatch(\tilde{\bidpreC}, \tilde{\bidC})=1$ and $\glive(\tilde{\bidC}, \envC)=1$, because $V_i$ never called $\gcapture$ in $\envC$ in this case. Thus, $\crec_i^{\c}$ is bad.

		\item \emph{Case 5: $\crec_i^{\r} = (\accept, \tilde{\vid}, \tilde{\ev})$.} 
			If $\crec_i^{\c}$ is not bad, then an entry $\ER_{i^{++}} = (\tilde{\vid},\dots)$ exists such that $\ev_{i^{++}}=\tilde{\ev}$. Since $\cvalid_i^{\c}=1$, $\tilde{\vid}=\vid'_{i^+}$ and $\tilde{\ev} = \ev_i^{\c}$. Since each $\vid$ in $\ER$ is distinct, $i^{++}=i^+$. Thus, $\ev_{i^{+}} = \ev_{i^{++}}=\ev_i^{\c}$. \qedhere
	\end{itemize}
\end{proof}

\begin{lemma}
	\label{lem:stuffing-implies-bad-items}
	If $\ver_{\prepareRoll}=1$, then for each $\ER_{i^{+}}$ entry s.t. $\ev_{i^{+}}$ is valid, \textbf{(1)} if the decryption $\bidR_{i^{+}}'$ of $c'_{\bidR_{i^+}}$ does not capture any voter in $(V_i)_{i \in I_{\ch} \cup I_{\cdh}}$ then $\ER_{i^{+}}$ is bad;
	\textbf{(2)} if $\bidR_{i^{+}}'$ captures some voter $V$ in $(V_i)_{i \in I_{\ch} \cup I_{\cdh}}$ but $V$ does not own $\ER_{i^{+}}$ then the $\ERR_{i^*}$ entry mapping to $\ER_{i^{+}}$ is bad.
\end{lemma}
\begin{proof}
	\textbf{(1)} Note that $\bidR_{i^{+}}'$ must capture some voter in $\Voters$, otherwise by Assumption \ref{assumption:CE}, $\gelg(\bidR'_{i^+},\ed)=0$ for any $\ed$ and thus $\ER_{i^+}$ is bad. If $\bidR_{i^{+}}'$ captures some voter $V_i \in \Voters$ but $i \not\in I_{\ch} \cup I_{\cdh}$, then the following cases arise:
	\begin{itemize}[leftmargin=*]
		\item[-] \emph{No $\ocast(i\in I_{\h}, \dots)$ or $\mathsf{OCapture}(i\not\in I_{\h}, \envC)$ call was made:} In this case, no $\gcapture(V_{i}, \envC)$ call was made. By Assumption \ref{assumption:UP}, the decryption $\bidC_{i^+}$ of $c_{\bidC_{i^+}}$ cannot satisfy $\gmatch(\bidR_{i^{+}}', \bidC_{i^+})=1$ and $\glive(\bidC_{i^+}, \envC)=1$. Thus, $\ER_{i^+}$ is bad.
		\item[-] \emph{Some $\mathsf{OCapture}(i\not\in I_{\h}, \envC)$ call was made but $\forall \ed: \gelg(p, \ed)=0$, where $p$ denotes the photograph capturing $V_i$ in this call:} Since both $\bidR'_{i^{+}}$ and $p$ capture $V_{i}$, $\gmatch(\bidR'_{i^{+}}$, $p)=1$ by Assumption \ref{assumption:CP}. Thus, by Assumption \ref{assumption:CE}, $\gelg(\bidR'_{i^{+}},\ed'_{i^{+}})=0$. Thus, $\ER_{i^+}$ is bad.
	\end{itemize}

	\noindent \textbf{(2)} Since $\bidR_{i^{+}}'$ captures $V_i$, the entry $\ERR_{i^*}$ mapping to $\ER_{i^{+}}$ is such that $\bidR_{i^*}$ captures $V_i$ too. If $V_i$ does not own $\ER_{i^{+}}$, then by Definition \ref{def:ownership}, $V_i$ does not own $\ERR_{i^*}$. Since $\bidR_{i^*}$ captures $V_i$, by Definition \ref{def:ownership}, $\guid(\bidR_{i^*}, \uid_{i^*}) \neq 1$. Thus, $\ERR_{i^*}$ is bad. \qedhere  
\end{proof} 

\begin{definition}[Hypergeometric distribution]
	Let there be a bag of $n$ objects out of which $f$ are special in some way. Let $\alpha$ random objects be sampled from the bag without replacement and let $d$ denote the probability of detecting a special object once chosen. Then, we let $\mathsf{Hyp}(n, \alpha, f, d)$ denote the probability of \emph{not} detecting even a single special object. By standard probability theory, $\mathsf{Hyp}(n, \alpha, f, 1) = \frac{\binom{n-f}{\alpha}}{\binom{n}{\alpha}}$ and $\mathsf{Hyp}(n, \alpha, f, d) = \sum\limits_{k=0}^{\alpha} \frac{\binom{f}{k}\binom{n-f}{\alpha-k}}{\binom{n}{\alpha}} \cdot (1-d)^k$ for $d<1$.
\end{definition}

\begin{theorem}[Soundness]
	\label{thm:soundness}
	Under the simulation security of the $\thsd$ protocol, our electoral roll protocol provides $\epsilon(n_{\d}$, $\alpha_{\d}$, $f_{\d}$, $N_{\s}$, $n_{\s}$, $\alpha_{\s}$, $f_{\s})$-soundness (see Definition \ref{def:soundness}), where $\epsilon(n_{\d}$, $\alpha_{\d}$, $f_{\d}$, $N_{\s}$, $n_{\s}$, $\alpha_{\s}$, $f_{\s}):=2\max(\{(\mathsf{Hyp}(n_{\d}$, $\alpha_{\d}$, $f_{\d}/2$, $1/3))^2$, $(\mathsf{Hyp}(N_{\s}+f_{\s}$, $\alpha_{\s}$, $f_{\s}/2$, $1))^2\}$, in the random oracle model.
\end{theorem}
\begin{proof}
	\renewcommand{\r}{\ensuremath{\mathsf{r}}}
	\renewcommand{\c}{\ensuremath{\mathsf{c}}}	
	Suppose there exists a PPT adversary $\adv$, parameters $\lambda$, $n_{\d}$, $\alpha_{\d}$, $f_{\d}$, $N_{\s}$, $n_{\s}$, $\alpha_{\s}$, $f_{\s}$  $\in$ $\Nat$ and oracles $\gset = \{\gcapture$, $\glive$, $\gmatch$, $\gelg$, $\guid\}$ satisfying Assumptions \ref{assumption:CP}-\ref{assumption:CI} such that $\Pr[\mathsf{Exp}_{\mathsf{soundness}}^{\adv,\gset}(1^{\lambda}$, $n_{\d}$, $\alpha_{\d}$, $f_{\d}$, $N_{\s}$, $n_{\s}$, $\alpha_{\s}$, $f_{\s})$ $=$ $1]$ $>$ $\epsilon(n_{\d}$, $\alpha_{\d}$, $f_{\d}$, $N_{\s}$, $n_{\s}$, $\alpha_{\s}$, $f_{\s})$. This leads to the following two cases (in both the cases, we assume that all decryptions obtained by auditor $A$ are correct - this holds under the simulation security of $\thsd$): \\

	\noindent \textbf{Case 1 [Voter denial]: $\Pr[\ver=1$ $\wedge$ $|I_{\ch}|=n_{\d}$ $\wedge$ $|\EVh \setminus \EV| > f_{\d}]$ $>$ $\epsilon(n_{\d}$, $\alpha_{\d}$, $f_{\d}$, $N_{\s}$, $n_{\s}$, $\alpha_{\s}$, $f_{\s})/2$}\textbf{.} As per $\mathsf{Exp}_{\mathsf{soundness}}$, each $\ev \in \EVh$ equals some $\ev_i^{\c}$ obtained in a unique $\ocast(i, \dots)$ call for some $i \in I_{\ch}$. Thus, $\EVh = \multiset{\ev_i^{\c} \mid i \in I_{\ch}}$. For each $\ev \in \EVh$, let $I_{\ev}:=\{i \in I_{\ch} \mid \ev_i^{\c} = \ev\}$ represent voters whose given encrypted vote was $\ev$ (multiple voters could be given the same $\ev$ by the adversary). Since each voter $(V_i)_{i \in I_{\ch}}$ is allowed to set $\ev_i^{\c}$ only once, each $I_{\ev}$ set is disjoint. Thus, we can write $|\EVh \setminus \EV| = \sum_{\ev \in \EVh} \max(I_{\ev}-\mult{\EV}{\ev},0)$, where $\mult{\EV}{\ev}$ equals the number of indices $i^{+}$ in list $\ve{\ev}$ produced by $\adv$ such that $\ev_{i^{+}} = \ev$ ($\mult{\EV}{\ev}=0$ if no such index exists). 
	
	By Lemma \ref{lem:ownership}, each $\ER$ entry is owned by at most one voter. Thus, at most $\mult{\EV}{\ev}$ voters own an $\ER_{i^{+}}$ entry such that $\ev_{i^{+}}=\ev$. Thus, the size of the set $I_{\ev,\mathsf{denied}}:=\{ i \in I_{\ev} \mid V_i$ does not own any $\ER_{i^{+}}$ entry s.t. $\ev_{i^{+}} = \ev  \}$ is at least $\max(|I_{\ev}|-\mult{\EV}{\ev}, 0)$. Since each $I_{\ev}$ set is disjoint, each $I_{\ev,\mathsf{denied}}$ set is disjoint too. Thus, the size of the set $I_{\mathsf{denied}} := \{ i \in I_{\ch} \mid V_i \text{ does not own any } \ER_{i^{+}} \text{ entry s.t. } \ev_{i^{+}} = \ev_i^{\c} \}$ $=$  $\bigcup_{\ev \in \EVh} I_{\ev,\mathsf{denied}}$ is at least $\sum_{\ev \in \EVh} \max(|I_{\ev}|-\mult{\EV}{\ev},0)$ $=$ $|\EVh \setminus \EV|$. Thus, $|\EVh \setminus \EV|>f$ implies $|I_{\mathsf{denied}}| > f$.

	By Lemma \ref{lem:denial-implies-bad-items}, if $i \in I_{\mathsf{denied}}$, i.e., if $V_i$ does not own an $\ER_{i^{+}}$ entry s.t. $\ev_{i^{+}}=\ev_i^{\c}$, then $V_i$ owns a bad receipt. Since by Lemma \ref{lem:ownership}, each item is owned by at most one voter, no bad receipt is owned by two different voters $V_{i_1},V_{i_2}$ for $i_1,i_2 \in I_{\mathsf{denied}}$. Thus, if $|I_{\mathsf{denied}}|>f$, then at least $f$ distinct bad receipts are produced. 
	
	The probability that $\ver=1$ in this case is at most $\max\limits_{f_1,f_2 \text{ s.t. } f_1+f_2=f_{\d}}$ $\{\mathsf{Hyp}(n_{\d}$, $\alpha_{\d}$, $f_1$, $1/3)$ $\cdot$ $\mathsf{Hyp}(n_{\d}$, $\alpha_{\d}$, $f_2$, $1)\}$, where the $1/3$ term comes from the fact that the probability of detection of a bad registration receipt is only $1/3$. This is at most $\max\limits_{f_1,f_2 \text{ s.t. } f_1+f_2=f_{\d}}$ $\{\mathsf{Hyp}(n_{\d}$, $\alpha_{\d}$, $f_1$, $1/3)$ $\cdot$ $\mathsf{Hyp}(n_{\d}$, $\alpha_{\d}$, $f_2$, $1/3)\}$ $=$ $(\mathsf{Hyp}(n_{\d}$, $\alpha_{\d}$, $f_{\d}/2$, $1/3))^2$, since by symmetry the expression takes the maximum value when $f_1=f_2=f_{\d}/2$. This is a contradiction. \\

	\noindent \textbf{Case 2 [Vote stuffing]: $\Pr[\ver=1$ $\wedge$ $|I_{\rh}|+|I_{\rdh}|=N_{\s}$ $\wedge$ $|I_{\ch}|+|I_{\cdh}|=n_{\s}$ $\wedge$ $|\EV| - |I_{\ch} \cup I_{\cdh}| > f_{\s}]$ $>$ $\epsilon(n_{\d}$, $\alpha_{\d}$, $f_{\d}$, $N_{\s}$, $n_{\s}$, $\alpha_{\s}$, $f_{\s})/2$}\textbf{.} Note that each element $\ev \in \EV$ appears in list $\ve{\ev}$ at some distinct index $i^{+}$ such that $\ER_{i^+}$ exists. Let $\bidR'_{i^{+}}$ denote the decryption of $c'_{\bidR_{i^{+}}}$ at $\ER_{i^+}$. Let $I_{\mathsf{extravoter}}^{+}:=\{ i^+ \mid \bidR'_{i^{+}}$  does not capture any voter in  $(V_i)_{i \in I_{\ch} \cup I_{\cdh}} \}$ and $I_{\mathsf{extravote},i}^{+}:=\{ i^{+} \mid \bidR'_{i^{+}}$ captures $V_i$ for some $i \in I_{\ch} \cup I_{\cdh}$ but $V_i$ does not own $\ER_{i^{+}} \}$. Since by Lemma \ref{lem:ownership}, each voter can own at most one $\ER$ entry, if there are $m$ indices $i^{+}$ such that $\bidR_{i^{+}}$ captures $V_i$ then $|I^{+}_{\mathsf{extravote},i}| \geq m-1$. Let $I_{\mathsf{stuffed}}^{+}:= I_{\mathsf{extravoter}}^{+} \cup \bigcup_{i \in I_{\ch} \cup I_{\cdh}} I_{\mathsf{extravote},i}^{+}$. Then $|I_{\mathsf{stuffed}}^{+}| = |\EV| - |I_{\ch} \cup I_{\cdh}|$. Thus, $|\EV| - |I_{\ch} \cup I_{\cdh}|>f_{\s}$ implies $|I_{\mathsf{stuffed}}^{+}|>f_{\s}$.

	By Lemma \ref{lem:stuffing-implies-bad-items}, for any $i^+ \in I_{\mathsf{extravoter}}^{+}$, $\ER_{i^+}$ is bad and for any $i^+ \in I_{\mathsf{extravote},i}^{+}$ for any $i$, the $\ERR_{i^*}$ entry mapping to $\ER_{i^+}$ is bad. Since by Lemma \ref{lem:correct-shuffle}, each $\ERR$ entry maps to only one $\ER$ entry, if $|I_{\mathsf{stuffed}}^{+}|>f_{\s}$, then at least $f_{\s}$ distinct bad ERR/ER entries are produced. 

	The probability that $\ver=1$ in this case is at most $\max\limits_{f_1,f_2 \text{ s.t. } f_1+f_2=f_{\s}}$ $\{\mathsf{Hyp}(N_{\s}+f_{\s}$, $\alpha_{\s}$, $f_1$, $1)$ $\cdot$ $\mathsf{Hyp}(N_{\s}+f_{\s}$, $\alpha_{\s}$, $f_2$, $1)\}$ $=$ $(\mathsf{Hyp}(N_{\s}+f_{\s}$, $\alpha_{\s}$, $f_{\s}/2$, $1))^2$, since by symmetry the expression takes the maximum value when $f_1=f_2=f_{\s}/2$. This is a contradiction. \qedhere
\end{proof}

\begin{theorem}[Privacy]
	\label{thm:privacy}
	Under the IND-CPA security of the $\ths$ scheme, our protocol provides $\delta(n,\alpha)$-privacy (see Definition \ref{def:privacy}), where $\delta(n,\alpha):=\max\{1$ $-$ $(\mathsf{Hyp}(n$, $\alpha$, $1$, $1))^2$ $(1$ $-$ $\alpha/n)^2$, $1$ $-$ $\mathsf{Hyp}(n$, $\alpha$, $2$, $1)$ $(1$ $-$ $2\alpha/n)$, $1$ $-$ $(\mathsf{Hyp}(n$, $\alpha$, $1$, $1))^2\}$, in the random oracle model.
\end{theorem}
\begin{proof}[Proof (Sketch)] \textbf{(1)} Let $E^{\mathsf{prof}}_b$ for $b \in \{0,1\}$ denote the two worlds in the voter profiling experiment. Let the event $\mathsf{Bad}$ be defined as when the $\ERR_{i^*}$ or $\ER_{i^+}$ entry owned by $V_i$ gets audited or when $V_i$ participates in $\indrraudit$ or $\indcraudit$. Clearly, $\Pr[\mathsf{Bad}] \leq \delta_1(n,\alpha)$, where $\delta_1(n,\alpha):=1 - (\mathsf{Hyp}(n,\alpha,1,1))^2(1-\alpha/n)^2$ (modelling $H_{\alpha}$ as a random oracle). Thus, $\adv$'s advantage in distinguishing between $E^{\mathsf{prof}}_0$ and $E^{\mathsf{prof}}_1$ is at most $\delta_1(n,\alpha) + \mathsf{adv}_{\neg \mathsf{Bad}}$, where $\mathsf{adv}_{\neg \mathsf{Bad}}$ denotes $\adv$'s advantage in case $\mathsf{Bad}$ does not happen. If $\mathsf{Bad}$ does not happen, $V_i$'s encrypted $\ERR$ or $\ER$ entries never get decrypted. This follows because
	\begin{itemize}[leftmargin=*] 
		\item[-] By the soundness of $\rho_{\mathsf{shuf}}$, the encryptions in $\ER$ are simply permutations and re-encryptions of the encryptions in $\ERR$ (implying in particular that the number of $\ER$ entries containing $V_i$'s information cannot be inflated).
		\item[-] During the $\register$ protocol, any voter $V \in \Voters$ obtains the decryption of any component of $\ERR_{i^*}$ only if $V$ is denied and a photograph $\bidpreR$ capturing $V$ matches the decryption $\bidR_{i^*}$ of $c_{\bidR_{i^*}}$ capturing $V_i$, implying that $V$ must be $V_i$ by Assumption \ref{assumption:NT}, which is a contradiction because honest voter $V_i$ does not get denied when $R_m$, $P_l$ and $E_j$ are honest. \label{enum:register-receipt-leak}
		\item[-] During the $\cast$ protocol, any voter $V \in \Voters$ obtains the decryption of any component of $\ER_{i^+}$ only if $V$ is denied and a photograph $\bidpreC$ capturing $V$ matches the decryption $\bidR'_{i^+}$ of $c'_{\bidR_{i^+}}$ capturing $V_i$, implying that $V$ must be $V_i$ by Assumption \ref{assumption:NT}, which is a contradiction as in case \ref{enum:register-receipt-leak}.
	\end{itemize}
	Further, all NIZKs issued for $V_i$ can be simulated. Thus, $\mathsf{adv}_{\neg \mathsf{Bad}}$ is negligible by the IND-CPA security of $\ths$. Thus, $\adv$'s advantage in distinguishing between $E^{\mathsf{prof}}_0$ and $E^{\mathsf{prof}}_1$ is at most negligibly more than $\delta_1(n,\alpha)$.
	
	\textbf{(2)} Let $E^{\mathsf{fa}}_b$ for $b \in \{0,1\}$ denote the two worlds in the forced abstention experiment. Let the event $\mathsf{Bad}$ be defined as when the $\ER_{i_0^+}$ entry owned by $V_{i_0}$ or the $\ER_{i_1^+}$ entry owned by $V_{i_1}$ gets audited or when $V_{i_0}$ or $V_{i_1}$ participates in the $\indcraudit$ protocol. Proceeding as above, $\adv$'s advantage in distinguishing between $E^{\mathsf{fa}}_0$ and $E^{\mathsf{fa}}_1$ is at most $\delta_2(\alpha, n) + \mathsf{adv}_{\neg \mathsf{Bad}}$, where $\delta_2(\alpha,n):=1 - \mathsf{Hyp}(n,\alpha,2,1)(1-2\alpha/n)$. Also, as above, if $\mathsf{Bad}$ does not happen, entries $\CI_{i_0^+}$ and $\CI_{i_1^+}$ never get decrypted. Further, since $P_l$ uploads dummy $\ev$ and $c_{\rho_{\ev}}$ for each $\ER$ entry for which the voter did not show up, $\mathsf{adv}_{\neg \mathsf{Bad}}$ is negligible by the IND-CPA security of $\ths$. (Note that we assume that the dummy $\ev$ given by $\pietoev$ is indistinguishable from a non-dummy one - see \S\ref{sec:protocol-structure}). Thus, $\adv$'s advantage in distinguishing between $E^{\mathsf{fa}}_0$ and $E^{\mathsf{fa}}_1$ is at most negligibly more than $\delta_2(\alpha, n)$.
	
	\textbf{(3)} Let $E^{\mathsf{ul}}_b$ for $b \in \{0,1\}$ denote the two worlds in the unlinkability experiment. Let the event $\mathsf{Bad}$ be defined as when the $\ERR_{i^*}$ or $\ER_{i^+}$ entry owned by $V_{i}$ gets audited. Proceeding as above, $\adv$'s advantage in distinguishing between $E^{\mathsf{ul}}_0$ and $E^{\mathsf{ul}}_1$ is at most $\delta_3(n,\alpha) + \mathsf{adv}_{\neg \mathsf{Bad}}$, where $\delta_3(n,\alpha):=1 - (\mathsf{Hyp}(n,\alpha,1,1))^2$. If $\mathsf{Bad}$ does not happen, as above, entries $\ERR_{i^*}$, $\ER_{i^+}$ and $\CI_{i^+}$ never get decrypted. Further, the proof of shuffle $\rho_{\mathsf{shuf}}$ can be simulated since the secret permutation of the honest $S_{k^*}$ is uniformly randomly chosen. Below, we argue that $V_i$'s participation in $\indrraudit$ or $\indcraudit$ protocols does not give $\adv$ any additional advantage. 
	
	Note that since $V_{i}$ is honest and all $R_m$, $E_j$ and $P_l$ are honest, $V_{i}$ always gets $\accept$ registration and cast receipts. Also, the $\accept$ cast receipt does not give any additional advantage to $\adv$ by the indistinguishability of $\ev$ obtained from the underlying E2E-V protocol. Thus, we focus on the $\accept$ registration receipt. 
	
	First, if $V_i$ chooses $\beta=2$ during $\indrraudit$ (gives the right receipt half), then $\adv$'s advantage is negligible by the IND-CPA security of $\ths$ and the zero-knowledgeness of $\rho_{\mathsf{eq}}$. If $V_i$ chooses $\beta=0$ (gives the top half), then by the IND-CPA security of $\ths$ and the information theoretically secure secret-sharing of $\rho_{\mathsf{eq}}$, $\adv$'s advantage is negligibly close to the case where $\adv$ only obtains $\uid$ chosen by $V_i$ (note that $\bidR$ in the receipt must be such that $\guid(\bidR, \uid)=1$ and thus it could be simulated using $\uid$, say given $V_i$'s identity card). Similarly, if $V_i$ chooses $\beta=1$ (gives the bottom half), $\adv$'s advantage is negligibly close to the case where $\adv$ only obtains $\ed,j$ chosen by $V_i$.
	
	Note that there are four uniformly random choices for $V_i$'s input: $C_0 = (\uid_0, \ed_0, j_0)$, $C_1 = (\uid_0, \ed_1, j_1)$, $C_2 = (\uid_1, \ed_0, j_0)$ and $C_3 = (\uid_1, \ed_1, j_1)$, where $C_0$ and $C_3$ correspond to the first world and $C_1$ and $C_2$ correspond to the second world. If $\adv$ obtains $\uid$, it needs to decide whether $C_0$ or $C_1$ was chosen if $\uid=\uid_0$ and whether $C_2$ or $C_3$ was chosen if $\uid=\uid_1$. If $\adv$ obtains $\ed,j$, it needs to decide whether $C_0$ or $C_2$ was chosen if $(\ed,j)=(\ed_0,j_0)$ and whether $C_1$ or $C_3$ was chosen if $(\ed,j)=(\ed_1,j_1)$. $\adv$'s advantage in all these decisions is zero.
	
	Thus, $\adv$'s advantage in distinguishing between $E^{\mathsf{ul}}_0$ and $E^{\mathsf{ul}}_1$ is at most negligibly more than $\delta_3(n,\alpha)$.
\end{proof}	

\section{Practicalities}
\label{sec:pracs}

We wrote a small computer program (available at \citep{artifact}) to estimate $\epsilon(n_{\d}$, $\alpha_{\d}$, $f_{\d}$, $N_{\s}$, $n_{\s}$, $\alpha_{\s}$, $f_{\s})$ and $\delta(n,\alpha)$ reported in Theorems \ref{thm:soundness} and \ref{thm:privacy} for real elections. We consider a realistic large election with both honest eligible casting voters $n_{\d}=n$ and total eligible casting voters $n_{\s}$ such that $n,n_{\d}, n_{\s} \approx 10^6$ (a million voters) and $n_s > N_s/2$ (i.e., more than half the registered eligible voters cast their vote). If the reported winning margin is $>2\%$, then assuming all of it to be due to eligibility frauds, either $f_{\s}>0.01n_{\s}$ or $f_{\d}>0.01n_{\d}$. Then, as an example, for $\alpha_{\s}, \alpha_{\d}, \alpha =2500$, we get $\epsilon(n_{\d}$, $\alpha_{\d}$, $f_{\d}$, $N_{\s}$, $n_{\s}$, $\alpha_{\s}$, $f_{\s}) < 0.0005$ and $\delta(n,\alpha)<0.01$. The $\epsilon$ estimate means that by verifying only $2500/10^6 = 0.25\%$ random voters' receipts and ERR/ER entries, the probability of not detecting any eligibility fraud is less than $0.0005$. The $\alpha$ value also means that less than 2500 manual checks, as part of $\gelg$, $\gmatch$ and $\glive$ oracles, are required to achieve this guarantee. The $\delta$ estimate can be interpreted to say that the soundness guarantee comes at the cost of leaking the identity, eligibility or participation information of at most $1\%$ random voters.

We also implemented our entire protocol (Figures \ref{fig:protocol-register}-\ref{fig:protocol-cast}) to estimate its performance (also available at \citep{artifact}). We used dummy oracles here that satisfy Assumptions \ref{assumption:CP}-\ref{assumption:CI} but do not actually implement any face matching, etc., because these checks would likely be manual in a real deployment. We used the threshold El Gamal scheme \citep{threshold-cryptography} to implement $\ths$. Although photographs may not fit within the message space of this $\ths$, secret-sharing actual images among the backend servers, indexing them by a collision-resistant hash and encrypting only the hash fixes this. We used the commitment-consistent proof of shuffle techniques of \citep{commitment-consistent-proof-shuffle} and \citep{terelius-restricted-shuffles} for implementing $\rho_{\mathsf{shuf}}$ during $\prepareRoll$. We also implemented batching techniques of \citep{bellare-batch-verification} for efficiently verifying proofs of correct threshold decryption of $\ve{c_{\vid}}'$ and $\ve{c_j}'$.

All our experiments ran on a single-core Intel Xeon W-1270 CPU with clock speed 3.40GHz and 64 GB RAM. For an election with $n = 10^6$ voters, the $\prepareRoll$ protocol, which is the most-expensive step, takes approximately 15 hours for both the auditor and each backend server. Given that the $\prepareRoll$ step happens after the registration closing deadline before polling, such delay is acceptable for large elections. Also, parallelisation opportunities exist, which are expected to bring this time to within an hour.	$\register$ and $\cast$ are instantaneous for each voter. Receipts contain less than $\approx~1.65$ KB machine-readable data (modulo the remark about storing photographs above), which can fit in standard QR codes.

\section{Conclusion}

We identified security and privacy threats with electoral rolls and polling-booth eligibility verification processes and formalised their security requirements under a non-cryptographic notion of eligibility. We then presented a concrete electoral roll protocol providing practical protection in large elections against $a)$ registration and polling-stage eligibility frauds and $b)$ voter profiling, linking and forced abstention attacks. Our guarantees crucially depend upon a trusted deduplicated primary identity system. Making such systems publicly verifiable would be the next big challenge in achieving truly secure electoral rolls deployable for real elections.

\section*{Acknowledgments}

Prashant Agrawal is supported by Pankaj Jalote Doctoral Grant and Pankaj Gupta Chair in Privacy and Decentralisation.

\balance

\bibliography{electoralrolls}

\begin{thebibliography}{10}
\providecommand{\url}[1]{#1}
\csname url@samestyle\endcsname
\providecommand{\newblock}{\relax}
\providecommand{\bibinfo}[2]{#2}
\providecommand{\BIBentrySTDinterwordspacing}{\spaceskip=0pt\relax}
\providecommand{\BIBentryALTinterwordstretchfactor}{4}
\providecommand{\BIBentryALTinterwordspacing}{\spaceskip=\fontdimen2\font plus
\BIBentryALTinterwordstretchfactor\fontdimen3\font minus
  \fontdimen4\font\relax}
\providecommand{\BIBforeignlanguage}[2]{{%
\expandafter\ifx\csname l@#1\endcsname\relax
\typeout{** WARNING: IEEEtran.bst: No hyphenation pattern has been}%
\typeout{** loaded for the language `#1'. Using the pattern for}%
\typeout{** the default language instead.}%
\else
\language=\csname l@#1\endcsname
\fi
#2}}
\providecommand{\BIBdecl}{\relax}
\BIBdecl

\bibitem{pretavoter}
P.~Y.~A. Ryan, D.~Bismark, J.~Heather, S.~Schneider, and Z.~Xia, ``{Pr\^{e}T
  \`{a} Voter: a voter-verifiable voting system},'' \emph{Tr. Information
  Forensics and Security}, vol.~4, no.~4, pp. 662--673, 2009.

\bibitem{scantigrity}
D.~{Chaum}, A.~{Essex}, R.~{Carback}, J.~{Clark}, S.~{Popoveniuc},
  A.~{Sherman}, and P.~{Vora}, ``{Scantegrity: end-to-end voter-verifiable
  optical-scan voting},'' \emph{IEEE S\&P}, vol.~6, no.~3, pp. 40--46, 2008.

\bibitem{punchscan}
A.~Essex and J.~Clark, ``{Punchscan in practice: an E2E election case study},''
  in \emph{WOTE}, 2007.

\bibitem{scratchandvote}
B.~Adida and R.~L. Rivest, ``{Scratch \& Vote: self-contained paper-based
  cryptographic voting},'' in \emph{WPES}, 2006, pp. 29--40.

\bibitem{bingovoting}
J.-M. Bohli, J.~M\"{u}ller-Quade, and S.~R\"{o}hrich, ``{Bingo voting: secure
  and coercion-free voting using a trusted random number generator},'' in
  \emph{EVOTE-ID}, 2007, pp. 111--124.

\bibitem{starvote}
S.~Bell, J.~Benaloh, M.~D. Byrne, D.~Debeauvoir, B.~Eakin, P.~Kortum,
  N.~McBurnett, O.~Pereira, P.~B. Stark, D.~S. Wallach, G.~Fisher, J.~Montoya,
  M.~Parker, and M.~Winn, ``{STAR-vote: A Secure, Transparent, Auditable, and
  Reliable voting system},'' in \emph{EVT/WOTE}, 2013.

\bibitem{election-rigging-how-to-fight}
D.~Calingaert, ``{Election rigging and how to fight it},'' \emph{{Journal of
  Democracy}}, vol.~17, no.~3, pp. 138--151, 2006.

\bibitem{elections-without-democracy}
A.~Schedler, ``{Elections without democracy: The menu of manipulation},''
  \emph{{Journal of Democracy}}, vol.~13, no.~2, pp. 36--50, 2002.

\bibitem{how-to-rig-an-election}
N.~Cheeseman and B.~Klaas, \emph{{How to Rig an Election}}.\hskip 1em plus
  0.5em minus 0.4em\relax Yale University Press, 2018.

\bibitem{retnakumar-electoral-rolls-india}
J.~Retnakumar, ``{How Far are the Electoral Rolls in India Ideal for a
  Democracy?}'' \emph{Journal of South Asian Development}, vol.~4, no.~2, pp.
  137--160, 2009.

\bibitem{kodali-electoralroll-deletion}
\BIBentryALTinterwordspacing
S.~Kodali, ``{Tech Is Helping Ruling Parties Create 'Perfect Electoral Rolls'
  as the EC Looks on},'' The Wire, 2023. [Online]. Available:
  \url{https://thewire.in/tech/electoral-rolls-election-commission-manipulation}
\BIBentrySTDinterwordspacing

\bibitem{bhatnagar-electoralroll-deletion}
\BIBentryALTinterwordspacing
G.~V. Bhatnagar, ``{Lakhs of Voters Deleted Without Proper Verification in
  Andhra, Telangana},'' The Wire, 2019. [Online]. Available:
  \url{https://thewire.in/rights/lakhs-of-voters-deleted-without-proper-verification-in-andhra-telangana}
\BIBentrySTDinterwordspacing

\bibitem{special-issue-electoral-fraud-india-pak}
N.~Martin and D.~Picherit, ``{Special Issue: Electoral Fraud and Manipulation
  in India and Pakistan},'' \emph{Commonwealth \& Comparative Politics},
  vol.~58, no.~1, pp. 1--20, 2020.

\bibitem{RER60}
\BIBentryALTinterwordspacing
``{The Registration of Electoral Rules, 1960},'' Election Commission of India,
  1960. [Online]. Available:
  \url{https://old.eci.gov.in/files/file/15146-the-registration-of-electors-rules-1960/}
\BIBentrySTDinterwordspacing

\bibitem{howard-kreiss-survey-aus-can-uk-us}
P.~N. Howard and D.~Kreiss, ``{Political Parties and Voter Privacy: Australia,
  Canada, the United Kingdom, and United States in Comparative Perspective},''
  \emph{First Monday}, vol.~15, no.~12, 2010.

\bibitem{onselen-electoral-databases}
P.~Onselen and W.~Errington, ``{Electoral Databases: Big Brother or Democracy
  Unbound?}'' \emph{Australian Journal of Political Science}, vol.~39, no.~2,
  pp. 349--366, 2004.

\bibitem{errington-suiting-themselves}
W.~Errington and P.~van Onselen, ``{Suiting Themselves: Major Parties,
  Electoral Databases and Privacy},'' \emph{Australasian Parliamentary Review},
  vol.~20, no.~1, pp. 21--33, 2005.

\bibitem{hunter}
C.~Hunter, ``{Political privacy and online politics: How e-campaigning
  threatens voter privacy},'' \emph{First Monday}, vol.~7, no.~2, 2002.

\bibitem{voter-privacy-big-data}
I.~S. Rubinstein, ``{Voter privacy in the age of big data},'' \emph{Wisconsin
  Law Review}, p. 861, 2014.

\bibitem{chaturvedi-religion-predictor}
R.~Chaturvedi and S.~Chaturvedi, ``It's all in the name: A character based
  approach to infer religion,'' \emph{arXiv preprint arXiv:2010.14479}, 2020.

\bibitem{watt-borger}
N.~Watt and J.~Borger, ``{Tories reveal secret weapon to target voters},'' The
  Guardian, 2004, [Saturday 9th October].

\bibitem{toi-electoralroll-fake-deletion-requests}
\BIBentryALTinterwordspacing
S.~Aluri, ``{EC gets 13 lakh deletion of vote pleas from Andhra Pradesh},'' The
  Times of India, 2019. [Online]. Available:
  \url{https://timesofindia.indiatimes.com/city/vijayawada/ec-gets-13-lakh-deletion-of-vote-pleas-from-andhra-pradesh/articleshow/68266523.cms}
\BIBentrySTDinterwordspacing

\bibitem{juels-coercion-resistance}
A.~Juels, D.~Catalano, and M.~Jakobsson, ``Coercion-resistant electronic
  elections,'' in \emph{Proceedings of the 2005 ACM Workshop on Privacy in the
  Electronic Society}, 2005, pp. 61--70.

\bibitem{araujo-coercion-resistance}
R.~Ara{\'u}jo, N.~B. Rajeb, R.~Robbana, J.~Traor{\'e}, and S.~Youssfi,
  ``Towards practical and secure coercion-resistant electronic elections,'' in
  \emph{International Conference on Cryptology and Network Security}.\hskip 1em
  plus 0.5em minus 0.4em\relax Springer, 2010, pp. 278--297.

\bibitem{spycher-coercion-resistance-linear}
O.~Spycher, R.~Koenig, R.~Haenni, and M.~Schl{\"a}pfer, ``A new approach
  towards coercion-resistant remote e-voting in linear time,'' in
  \emph{International Conference on Financial Cryptography and Data
  Security}.\hskip 1em plus 0.5em minus 0.4em\relax Springer, 2011, pp.
  182--189.

\bibitem{essex-cobra}
A.~Essex, J.~Clark, and U.~Hengartner, ``Cobra: Toward concurrent ballot
  authorization for internet voting.'' \emph{EVT/WOTE}, vol.~12, 2012.

\bibitem{clark-selections}
J.~Clark and U.~Hengartner, ``Selections: Internet voting with
  over-the-shoulder coercion-resistance,'' in \emph{International Conference on
  Financial Cryptography and Data Security}.\hskip 1em plus 0.5em minus
  0.4em\relax Springer, 2011, pp. 47--61.

\bibitem{araujo-coercion-resistance-fix}
R.~Ara{\'u}jo and J.~Traor{\'e}, ``A practical coercion resistant voting scheme
  revisited,'' in \emph{International Conference on E-Voting and
  Identity}.\hskip 1em plus 0.5em minus 0.4em\relax Springer, 2013, pp.
  193--209.

\bibitem{ktv-helios}
O.~Kulyk, V.~Teague, and M.~Volkamer, ``Extending helios towards private
  eligibility verifiability,'' in \emph{International Conference on E-Voting
  and Identity}.\hskip 1em plus 0.5em minus 0.4em\relax Springer, 2015, pp.
  57--73.

\bibitem{artifact}
\BIBentryALTinterwordspacing
``{Supplementary material: Publicly auditable privacy-preserving electoral
  rolls},'' 2024. [Online]. Available:
  \url{https://github.com/agrawalprash/secure-electoral-roll}
\BIBentrySTDinterwordspacing

\bibitem{faceflashing}
D.~Tang, Z.~Zhou, Y.~Zhang, and K.~Zhang, ``{Face flashing: a secure liveness
  detection protocol based on light reflections},'' \emph{arXiv preprint
  arXiv:1801.01949}, 2018.

\bibitem{facecloseup}
Y.~Li, Z.~Wang, Y.~Li, R.~Deng, B.~Chen, W.~Meng, and H.~Li, ``{A closer look
  tells more: a facial distortion based liveness detection for face
  authentication},'' in \emph{Proceedings of the 2019 ACM Asia Conference on
  Computer and Communications Security}, 2019, pp. 241--246.

\bibitem{facelive}
Y.~Li, Y.~Li, Q.~Yan, H.~Kong, and R.~H. Deng, ``{Seeing your face is not
  enough: An inertial sensor-based liveness detection for face
  authentication},'' in \emph{Proceedings of the 22nd ACM SIGSAC Conference on
  Computer and Communications Security}, 2015, pp. 1558--1569.

\bibitem{BB0k}
J.~Heather and D.~Lundin, ``{The append-only web bulletin board},'' in
  \emph{Formal Aspects in Security and Trust}, P.~Degano, J.~Guttman, and
  F.~Martinelli, Eds., 2009, pp. 242--256.

\bibitem{threshold-cryptography}
Y.~Desmedt, ``{Threshold cryptography},'' \emph{European Tr.
  Telecommunications}, vol.~5, no.~4, pp. 449--458, 1994.

\bibitem{commitment-consistent-proof-shuffle}
D.~Wikstr{\"o}m, ``{A commitment-consistent proof of a shuffle},'' in
  \emph{Australasian Conf. Information Security and Privacy}, 2009, pp.
  407--421.

\bibitem{terelius-restricted-shuffles}
B.~Terelius and D.~Wikstr{\"o}m, ``{Proofs of restricted shuffles},'' in
  \emph{Intl. Conf. Cryptology in Africa}, 2010, pp. 100--113.

\bibitem{mixnet-sok}
T.~Haines and J.~M{\"u}ller, ``{SoK: techniques for verifiable mix nets},'' in
  \emph{CSF}, 2020, pp. 49--64.

\bibitem{bulletin-board}
P.~Agrawal, S.~Sharma, and S.~Banerjee, ``Blockchain vs public bulletin board
  for integrity of elections and electoral rolls,'' \emph{The India Forum},
  2021.

\bibitem{bellare-batch-verification}
M.~Bellare, J.~A. Garay, and T.~Rabin, ``{Fast batch verification for modular
  exponentiation and digital signatures},'' in \emph{EUROCRYPT}, 1998, pp.
  236--250.

\end{thebibliography}
\bibliographystyle{IEEEtran}

\end{document}